\documentclass[letterpaper,twocolumn,10pt]{article}
\usepackage{usenix}

\usepackage{tikz}
\usepackage{amsmath}
\usepackage{graphicx} 
\usepackage{subfigure}
\usepackage{CJKutf8}
\usepackage{amsthm}
\usepackage{amssymb}
\usepackage{amsfonts}
\usepackage{amsmath}
\usepackage{cleveref}
\usepackage{booktabs}
\usepackage{multirow}
\usepackage{xspace}
\usepackage{bbding}

\usepackage{pifont}

\usepackage{tikz}

\newtheorem{definition}{Definition}[section]
\newtheorem{lemma}{Lemma}
\newtheorem{theorem}{Theorem}

\newcommand{\AHE}{\text{RHE}\xspace}

\newcommand{\skEn}{\mathtt{SkEnc}\xspace}
\newcommand{\pkEn}{\mathtt{PkEnc}\xspace}
\newcommand{\De}{\mathtt{Dec}\xspace}
\newcommand{\KeyGen}{\mathtt{KeyGen}\xspace}
\newcommand{\EvalAdd}{\mathtt{EvalAdd}\xspace}

\newcommand{\pkGen}{\mathtt{PkGen}\xspace}
\newcommand{\Setup}{\mathtt{Setup}\xspace}

\newcommand{\EvalMult}{\mathtt{EvalMult}\xspace}
\newcommand{\DeT}{\mathtt{Dec3}\xspace}
\newcommand{\EvkGen}{\mathtt{EvkGen}\xspace}
\newcommand{\Resize}{\mathtt{Resize}\xspace}

\newcommand{\negl}{\text{negligible}\xspace}

\newcommand{\wKeyGen}{\mathtt{WKeyGen}\xspace}
\newcommand{\Embed}{\mathtt{Embed}\xspace}
\newcommand{\Extract}{\mathtt{Extract}\xspace}
\newcommand{\emk}{\mathtt{emk}\xspace}
\newcommand{\exk}{\mathtt{exk}\xspace}

\newcommand{\AddRWatermark}{\text{ARWMark}\xspace}
\newcommand{\MultRWatermark}{\text{MRWMark}\xspace}

\newcommand{\sk}{\mathtt{sk}\xspace}
\newcommand{\pk}{\mathtt{pk}\xspace}
\newcommand{\evk}{\mathtt{evk}\xspace}

\newcommand{\skscheme}{\text{SkRHE}\xspace}
\newcommand{\pkscheme}{\text{PkRHE}\xspace}
\newcommand{\var}{\mathbf{D}\xspace}
\newcommand{\E}{\mathbf{E}\xspace}
\newcommand{\cov}{\mathbf{Cov}\xspace}

\crefname{table}{Table}{Tables}
\crefname{definition}{Def.}{Defs.}
\crefname{theorem}{Thm.}{Thms.}
\crefname{equation}{Eq.}{Eq.s}
\crefname{figure}{Fig.}{Figs.}
\crefname{appendix}{Appendix}{Appendices}
\crefname{algorithm}{Algorithm}{Algs.}

\usepackage{filecontents}

\begin{document}

\date{}

\title{Ciphertext-Native Watermarking for RLWE-Based Homomorphic Encryption}

\author{
{Yufei Zhou}\\
Sun Yat-sen University
\and
{Peijia Zheng}\\
Sun Yat-sen University
} 

\maketitle

\begin{abstract}
Homomorphic encryption (HE) schemes based on the Ring Learning with Errors (RLWE) problem have been rapidly developed and widely applied to secure computation tasks, such as privacy-preserving deep learning inference and database queries. However, existing HE schemes mainly focus on the feasibility and efficiency of homomorphic computation, while practical requirements including ciphertext copyright protection, provenance tracking, and computation supervision remain largely unexplored. In this work, we propose a watermarking technique for RLWE-based HE ciphertexts. By exploiting the algebraic structure of RLWE polynomials, we embed watermark information into ciphertext noise without affecting plaintext correctness. To address watermark degradation caused by homomorphic operations, we introduce two practical schemes. The first scheme, \AddRWatermark, leverages noise stratification to achieve robustness against homomorphic additive operations. The second scheme, \MultRWatermark, is based on the roots of a linear equation and supports zero-bit watermarking while remaining robust against both homomorphic addition and multiplication. We provide rigorous theoretical analysis demonstrating that the proposed schemes preserve the original security of HE while maintaining correctness and watermark robustness. Extensive experiments further validate the effectiveness and practicality of the proposed watermarking schemes.
\end{abstract}


\section{Introduction}
Homomorphic encryption (HE) enables computation directly over encrypted ciphertexts, thereby supporting non-interactive data processing~\cite{armknecht2015guide}. 
HE has been widely applied to privacy-preserving tasks, including medical data processing~\cite{lee2025comprehensive,yanez2026homomorphic} and natural language processing~\cite{moon2025thor,de2025encryptedllm}. 
Unlike non-homomorphic ciphertexts such as AES~\cite{rijmen2001advanced}, processed HE ciphertexts often embody additional value beyond the original encrypted data.
For example, after homomorphic image processing~\cite{challa2015secure}, the resulting ciphertext contains not only the original image information, but also the embedded processing logic and the computational effort contributed by the service provider. 
In this sense, homomorphic image processing adds value to the original encrypted image. 
Similarly, in privacy-preserving deep learning inference~\cite{moon2025thor,de2025encryptedllm}, encrypted inputs are evaluated by encrypted models to produce encrypted inference results. 
The resulting ciphertext is the output of the deployed model and therefore embodies the value of the model computation. 

Based on these observations, we argue that HE ciphertexts themselves are valuable digital assets. 
During transmission and computation, HE ciphertexts are continuously transformed and enriched, while their data flow involves multiple parties with potentially conflicting interests. 
As a result, HE ciphertexts require dedicated protection mechanisms. 
Unfortunately, existing HE schemes do not consider the protection of ciphertext ownership and associated rights.

For some applications, homomorphic hash functions~\cite{yao2018homomorphic}
or homomorphic message authentication codes (MACs)~\cite{feng2022multi}
are sufficient. However, these approaches require additional
authentication information that is maintained separately from the
original ciphertexts, introducing extra management overhead and
potential security concerns. Moreover, they are not designed to preserve
authentication information under general homomorphic evaluations. Noise
introduced during homomorphic computation, such as relinearization after
ciphertext multiplication~\cite{mono2023new} or approximate arithmetic
in CKKS~\cite{cheon2017homomorphic}, may invalidate the verification
process.

Similarly, zero-knowledge proofs (ZKPs)~\cite{bunz2018bulletproofs,yang2022non}
can verify the correctness of outsourced computations without revealing
private information. However, ZKPs mainly focus on final result
verification and cannot establish the provenance of intermediate
ciphertext states during homomorphic evaluation. For example, in a
multi-party homomorphic computation pipeline, a participant cannot prove
that its encrypted contribution was included in the computation, since
the intermediate ciphertexts remain encrypted throughout the process.
Therefore, existing approaches cannot provide ciphertext-level
provenance and accountability.

To address these limitations, we introduce watermarking techniques into
homomorphic ciphertexts to enable ciphertext provenance, ownership
attribution, and accountability. Unlike external authentication
mechanisms, our approach embeds watermark information directly into
ciphertexts and enables verification after homomorphic computations.
The embedded watermark remains transparent to the evaluator and does not
require any modification to the standard HE evaluation procedure.

However, embedding watermarks into homomorphic ciphertexts is
non-trivial, and existing watermarking techniques cannot be directly
applied to HE ciphertexts. HE ciphertexts follow strict mathematical
structures, and directly modifying ciphertext coefficients may cause
decryption failure. Even slight perturbations may corrupt the underlying
plaintext information, requiring watermark embedding to preserve
ciphertext correctness.

Moreover, HE ciphertexts are computationally indistinguishable from
uniform distributions, making traditional statistical watermarking
methods, such as difference expansion~\cite{peng2022reversible} and
histogram shifting~\cite{coatrieux2012reversible}, unsuitable for this
setting. In addition, homomorphic operations differ fundamentally from
transformations in conventional watermarking scenarios, and embedded
watermarks may be destroyed during homomorphic evaluation. In
particular, numerical relationships within or across ciphertexts can
change significantly after homomorphic addition or multiplication.

These challenges motivate us to design watermarking schemes that are
specifically tailored for homomorphic ciphertexts while preserving
decryption correctness, HE security, and robustness against
homomorphic operations.

Our contributions are summarized as follows.

\begin{itemize}

\item We formalize the problem of watermarking homomorphic ciphertexts
and define its security requirements, including ciphertext provenance,
ownership attribution, transparency, and robustness against homomorphic
operations.

\item We propose a 1-bit robust watermarking scheme,
\AddRWatermark, for additively homomorphic ciphertexts. The scheme
exploits the redundancy in ciphertext noise to embed additive
perturbations and employs correlation-based detection for robust
watermark extraction. It can be extended to support $t$-bit watermarks
while maintaining robustness under homomorphic addition.

\item We propose a zero-bit watermarking scheme,
\MultRWatermark, for fully homomorphic ciphertexts. The construction
leverages the solution space of linear equation systems to embed
watermark information and remains robust against both homomorphic
addition and multiplication.

\item We provide theoretical analysis and extensive experimental
evaluation of the proposed schemes. The results demonstrate their
security, robustness, and efficiency under various homomorphic
operations.

\end{itemize}





The implementation is available at the github
\href{https://github.com/pahjastia/Watermark_for_HE}{link}.

\section{Related Work}

\begin{table*}[t]
\centering
\caption{Comparison between existing approaches and homomorphic ciphertext watermarking.}
\label{tab:related_work_comparison}
\resizebox{\linewidth}{!}{
\begin{tabular}{lcccccc}
\toprule
\textbf{Approach} & \textbf{Target} & \textbf{Provenance} & \textbf{Transparency} & \textbf{HE Compatible} & \textbf{Survive HE Ops.} & \textbf{Extra Data} \\\midrule

Homomorphic MAC~\cite{feng2022multi} & Computation & $\times$ & $\times$ & Partial & Limited & Tag \\

ZKP \cite{bunz2018bulletproofs,yang2022non}& Computation Result & $\times$ & $\times$ & Partial & $\times$ & Proof \\

Multimedia Watermarking \cite{hartung1999multimedia,cox2002digital} & Digital Media & $\checkmark$ & $\times$ & $\times$ & $\times$ & No \\

Neural Network Watermarking \cite{uchida2017embedding} & ML Models & $\checkmark$ & Model-dependent & $\times$ & Model-dependent & No \\

Cryptographic Watermarking \cite{barak2012possibility,hopper2007weak,cohen2016watermarking} 
& Cryptographic Functions & $\checkmark$ & Not Designed & Not Designed & Not Designed & No \\\midrule

\textbf{Ours} & \textbf{HE Ciphertexts} & $\checkmark$ & $\checkmark$ & $\checkmark$ & $\checkmark$ & No \\\bottomrule
\end{tabular}
}
\end{table*}

\subsection{Verifiable and Authenticated Computation}

Verifiable computation and authenticated computation enable clients to
verify the correctness of outsourced computations performed by
untrusted servers. Homomorphic MACs and authenticated data structures
provide efficient verification mechanisms by attaching authentication
information to data or computation states
\cite{feng2022multi}. However, such authentication information is
typically maintained separately from the encrypted data and cannot
provide ownership attribution or provenance tracking for ciphertexts.

Zero-knowledge proof (ZKP)-based approaches allow a prover to convince a
verifier of the correctness of a computation without revealing private
information
\cite{bunz2018bulletproofs,yang2022non}. Nevertheless, existing ZKP techniques mainly focus on verifying the
correctness of the computation result and do not address the provenance
of intermediate encrypted states during homomorphic evaluation.

In contrast, we aim to protect the ownership and provenance of
homomorphic ciphertexts by embedding watermark information directly
into ciphertexts and enabling post-computation verification. The
embedded watermark remains transparent to the evaluator and introduces
no additional requirements during homomorphic computation.

\subsection{Digital and Cryptographic Watermarking}

Digital watermarking has been extensively studied for ownership
protection and content tracing in multimedia and machine learning
domains. Traditional multimedia watermarking schemes embed
identification information into digital media while preserving
perceptual quality
\cite{hartung1999multimedia,cox2002digital}. Neural network
watermarking further extends this idea to model ownership protection by
embedding watermarks into model parameters or behaviors
\cite{uchida2017embedding}. However, these techniques rely on
redundant structures in media signals or neural networks and cannot be
directly applied to homomorphic ciphertexts with strict algebraic
constraints.

Cryptographic function watermarking studies how to embed marks into
cryptographic functions while preserving their functionality and
preventing unauthorized mark removal
\cite{barak2012possibility,hopper2007weak,cohen2016watermarking}.
However, watermarking cryptographic functions and watermarking
homomorphic ciphertexts address fundamentally different objects. A
cryptographic function is a reusable algorithmic object, whereas a
homomorphic ciphertext is an individual encrypted instance that
undergoes continuous algebraic transformations during computation.
Existing cryptographic watermarking techniques therefore cannot be
directly applied to homomorphic ciphertexts, which require additional
guarantees including decryption correctness, security preservation, and
robustness against homomorphic operations.

Our work introduces ciphertext-native watermarking for RLWE-based
HE. Unlike previous watermarking approaches, our
scheme embeds ownership information into ciphertext structures while
preserving plaintext correctness, HE security, and watermark
detectability after homomorphic evaluations.

\cref{tab:related_work_comparison} summarizes the differences
between existing approaches and our proposed homomorphic ciphertext
watermarking schemes.

\section{Preliminaries}
In this section, we first introduce the RLWE problem and then present a simplified RLWE-based HE scheme.

\subsection{Basic Notation}
We define the polynomial ring $\mathbf{R}=\mathbb{Z}(x)/(x^N+1)$, where $N$ is a power of 2. For an integer $q$, we use $\mathbf{R}_q$ to denote $\mathbf{R}/q\mathbf{R}$. In other words, the elements in $\mathbf{R}_q$ are all polynomials with coefficient modulus $q$ and polynomial modulus $x^N+1$, and have integer coefficients. 

We use lowercase bold letters for polynomials. All polynomial operations are performed over $\mathbf{R}_q$ by default. We use $[i]$ to denote the $i$-th coefficient of a polynomial or the $i$-th element of a vector, where $i$ starts from $1$ by default.
For example, we use $\mathbf{a}[i]$ to denote the $i$-th coefficient of the polynomial $\mathbf{a}$. 
The inner product of two ring elements $\mathbf{a}$ and $\mathbf{b}$ is $\langle\mathbf{a},\mathbf{b}\rangle=\sum_{i=1}^{N}\mathbf{a}[i]\cdot \mathbf{b}[i] \mod q$. If there is no ambiguity, we use $\sum_i$ to denote $\sum_{i=1}^{N}$.  


For a positive integer $q$ and any $\mathbf{a} \in \mathbf{R}$, $\mathbf{a} \mod q$ will reduce the coefficients of $\mathbf{a}$ into the range $(-\frac{q}{2},\frac{q}{2}]$.
We use $\Vert \mathbf{a} \Vert=\max_{i} \vert \mathbf{a}[i] \vert$ to denote the infinite norm of $\mathbf{a}$. In other words, $\Vert \mathbf{a} \Vert$ is the  maximum absolute value of the coefficients of $\mathbf{a}$.

We say a function $f(x)$ from the natural numbers to the non-negative real numbers is $\negl$ if for all constants $c$ there exists an $N_0$ such that for all integers $n>N_0$ it holds that $f(n)<n^{-c}$~\cite{katz2015introduction}. An event is considered impossible (or computationally infeasible) if its probability is $\negl$. Two distributions are indistinguishable if the probability of distinguishing them is no more than 1/2 plus a negligible function. 

\subsection{The Ring Learning with Errors Problem}
The RLWE problem was introduced in~\cite{lyubashevsky2010ideal}. Here we give a special and simplified version of RLWE which is widely used in the design of HE schemes~\cite{brakerski2014leveled,cheon2017homomorphic}.

\begin{definition}
\label{def:RLWE}
    Let $\chi_s, \chi_e$ be two distributions over $\mathbf{R}_q$. The RLWE problem is to distinguish the following two distributions: 

    First, choose a random element $\mathbf{s}$ from distribution $\chi_s$, $m$ uniformly random elements $\mathbf{a}_1,\cdots,\mathbf{a}_m$ from $\mathbf{R}_q$, and $m$ elements $\mathbf{e}_1,\cdots,\mathbf{e}_m$ from distribution $\chi_e$. Compute $\mathbf{b}_i=\mathbf{a}_i\mathbf{s}+\mathbf{e}_i$ and output $(\mathbf{a}_i,\mathbf{b}_i)$, where $ i=1,\cdots,m$.
    
    Second, choose $2m$ uniformly random elements $\mathbf{a}_i^\prime,\cdots,\mathbf{a}_m^\prime,\mathbf{b}_1^\prime,\cdots,\mathbf{b}_m^\prime$ from $\mathbf{R}_q$. Output  $(\mathbf{a}_i^\prime,\mathbf{b}_i^\prime)$, where $ i=1,\cdots,m$.
\end{definition}
The RLWE assumption is that the RLWE problem is hard. 
As proved  in~\cite{lyubashevsky2010ideal}, the well-established shortest vector problem (SVP) \cite{micciancio2002shortest} over ideal lattices can be reduced to it. In other words, the RLWE problem is at least as hard as the SVP problem over ideal lattices.

\subsection{Homomorphic Encryption}
\label{sec:homomorphic_encryption}
We first present the definitions of additive homomorphism and multiplicative homomorphism:
\begin{definition}
\label{def:HE}
    Let $m_1,m_2$ be two plaintexts, $\operatorname{Enc}(\cdot)$ be the encryption function and $\operatorname{Dec}(\cdot)$ be the decryption function. An encryption scheme is additive homomorphic when $\operatorname{Dec}(\operatorname{Enc}(m_1) + \operatorname{Enc}(m_2))=m_1+m_2$. Similarly, an encryption scheme is multiplicative homomorphic when  $\operatorname{Dec}(\operatorname{Enc}(m_1) \times \operatorname{Enc}(m_2))=m_1 \times m_2$.
\end{definition}

Next, we introduce the RLWE-based HE scheme \AHE. The scheme supports two ciphertext forms: secret-key encryption \skscheme\ and public-key encryption \pkscheme.
The detailed procedures of the secret-key encryption scheme \skscheme are described as follows:
\begin{itemize}
    \item $\Setup(\lambda,p)$: Given the security parameter $\lambda$ and plaintext modulus $p$, determine the polynomial ring $\mathbf{R}_q$, the error distribution $\chi_e$, and the secret-key distribution $\chi_s$. Typically, $\chi_e$ is a bounded discrete Gaussian distribution, while $\chi_s$ is the ternary distribution $\{-1,0,1\}^N$~\cite{albrecht2021homomorphic}. The ciphertext modulus $q$ is chosen such that $\gcd(p,q)=1$.
    
    \item $\KeyGen$: Sample $\mathbf{s}$ from $\chi_s$. Output the secret key $\sk=(1,\mathbf{s})$.
    
    \item $\skEn(\sk,\mathbf{m})$: The message $\mathbf{m} $ is an element in $\mathbf{R}_p$. Sample $\mathbf{a}$ uniformly from $\mathbf{R}_q$. Then sample $\mathbf{e}$ from error distribution $\chi_e$.
    Compute $\mathbf{b}=\mathbf{a}\cdot \mathbf{s}+p\mathbf{e}+\mathbf{m}$. Output the ciphertext $c_t=(\mathbf{b},-\mathbf{a})$.
    
    \item $\De(\sk,c_t)$: Let ciphertext $c_t=(\mathbf{c}_1,\mathbf{c}_2)$. Output the plaintext $\mathbf{m}=(\mathbf{c}_1+\mathbf{s}\cdot \mathbf{c}_2) \mod p$.
    
    \item $\EvalAdd(c_t,c_t^\prime)$: Let $c_t=(\mathbf{c}_1,\mathbf{c}_2)$ be the ciphertext of $\mathbf{m}_1$ and $c_t^\prime=(\mathbf{c}^\prime_1,\mathbf{c}^\prime_2)$ be the ciphertext of $\mathbf{m}_2$. Output the ciphertext $(\mathbf{c}_1+\mathbf{c}^\prime_1,\mathbf{c}_2+\mathbf{c}^\prime_2)$, which is a ciphertext of the sum of $\mathbf{m}_1$ and $\mathbf{m}_2$.
\end{itemize}

The details of \pkscheme, together with the security analysis of \AHE, are provided in Appendix~\ref{appendix:ahe}.

\section{Homomorphic Ciphertext Watermarking}






We define homomorphic ciphertext watermarking as the process of modifying a ciphertext to embed information associated with that ciphertext. The embedded information is referred to as a watermark. A homomorphic ciphertext watermarking scheme typically consists of the following three algorithms:
\begin{itemize}
    \item $\wKeyGen(\lambda)$: Generate an embedding key $\emk$ and an extracting key $\exk$ that satisfy the security parameter $\lambda$. Output the keys $\emk, \exk$.

    \item $\Embed(\emk, w, c_t)$: Embed the watermark $w$ into the ciphertext $c_t$ with embedding key $\emk$. Output the embedded ciphertext $c_t^\prime$.

    \item $\Extract(\exk, c_t)$: Extract the watermark from the ciphertext $c_t$ with the help of the extracting key $\exk$. Output the embedded watermark or false when no watermark is embedded into $c_t$. 
\end{itemize}

The embedding key $\emk$ contains all secret information required for watermark embedding, while the extraction key $\exk$ contains the secret information required for watermark extraction.
When the embedded watermark is a zero-bit watermark, $\Extract$ outputs either true or false, indicating whether the ciphertext contains a watermark. When the embedded watermark encodes one or multiple bits, $\Extract$ outputs the embedded watermark value or false.

We illustrate the workflow of homomorphic ciphertext watermarking in \cref{fig:watermark_flow}. The embedding algorithm $\Embed$ uses the embedding key $\emk$ to embed the watermark information $w$ into the ciphertext $c_{t,0}$, producing the watermarked ciphertext $c_{t,1}$. The ciphertext $c_{t,1}$ is then processed through a sequence of homomorphic operations or perturbed by additional noise, resulting in $c_{t,1}^\prime$.
The extraction algorithm $\Extract$ then uses the extraction key $\exk$ to recover the watermark $w^\prime$ from $c_{t,1}^\prime$.

Although the workflow is straightforward to describe, establishing a well-defined homomorphic ciphertext watermarking scheme requires clarifying the properties that such a scheme should satisfy. Next, we first review several representative examples of robust watermarking in other domains, and then derive the properties required for watermarking in homomorphic ciphertexts.

\textbf{Multimedia Watermarking. }
Multimedia watermarking is a technique for embedding information into multimedia content~\cite{hartung1999multimedia,cox2002digital}. The embedding process is generally required not to affect the quality of the original media. For example, image watermarking schemes~\cite{wan2022comprehensive} should not introduce noticeable distortion to the image after embedding. At the same time, the embedded information is expected to remain robust under standard image processing operations, such as geometric transformations (e.g., cropping, scaling, and rotation)~\cite{hu2020cover} and compression during transmission~\cite{zhang2024robust,chen2020jsnet}. Watermarking schemes for other multimedia carriers follow similar requirements~\cite{ahvanooey2020anitw,zhang2023m,luo2023dvmark}.

\textbf{Neural Network Watermarking. }
Neural network watermarking embeds information into neural network models~\cite{uchida2017embedding}. Although the parameters of the watermarked model may differ significantly from those of the original model, the model performance should not degrade substantially after watermark embedding. In practice, neural network models are often fine-tuned or pruned to better fit downstream tasks. Therefore, robustness requires that the embedded watermark remain detectable even after fine-tuning or minor pruning~\cite{zhang2021deep,yasui2022coded}.

\textbf{Cryptographic Functions Watermarking. }
Cryptographic function watermarking embeds information into cryptographic functions~\cite{barak2012possibility,hopper2007weak,cohen2016watermarking}. The embedding process should preserve the functionality of the original function. That is, the outputs of the watermarked function should be computationally indistinguishable from those of the original function on most inputs. Robustness further requires that, as long as the modified function remains indistinguishable from the watermarked one, the probability of removing the watermark is negligible~\cite{goyal2019watermarking,nishimaki2020equipping,kitagawa2024watermarking}.

From the above examples, we identify two fundamental properties of watermarking. First, watermark embedding should preserve the essential properties of the host object, which we refer to as fidelity. Second, the embedded watermark should remain robust against noise introduced by operations on the host object, which we denote as robustness. Therefore, homomorphic ciphertext watermarking must simultaneously satisfy both fidelity and robustness.
In addition, homomorphic ciphertext watermarking must preserve the security guarantees inherited from the underlying encryption scheme. 

\begin{figure}[!t]
\centering
\includegraphics[width=\linewidth]{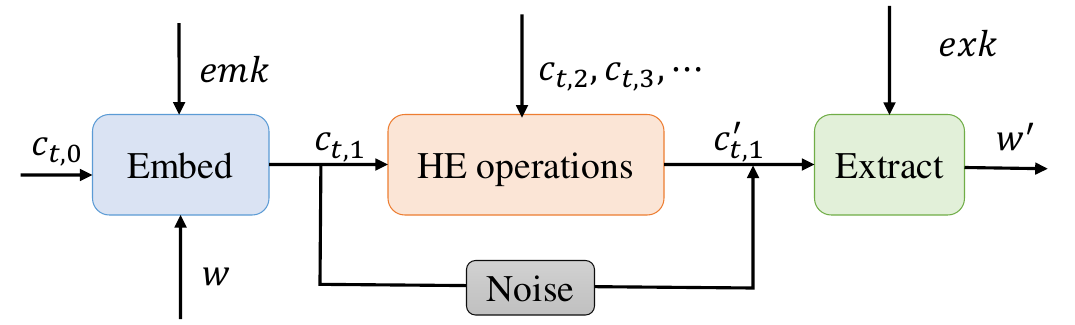}
\caption{Homomorphic ciphertext watermarking system. The HE operations usually are homomorphic additions or multiplications. Besides HE operations, other operations such as key switching may add noise to the ciphertext.}
\label{fig:watermark_flow}
\end{figure}




\subsection{Fidelity}
The property of fidelity for homomorphic ciphertext watermarking requires that watermark embedding does not alter the decryption result of the ciphertext. That is, the ciphertexts before and after watermark embedding should be equivalent with respect to decryption.
We formally define fidelity  as follows:
\begin{definition}
\label{def:fidelity}
    Let $c_t$ be a ciphertext of $\mathbf{m}$ encrypted with \AHE and $c_t^\prime$ be the ciphertext after embedding watermark into $c_t$ using a watermark scheme. Let $\sk$ be the secret key of \AHE. 
    The watermark scheme has fidelity when  $\De(\sk,c_t^\prime)=\De(\sk,c_t)$.
\end{definition}

\subsection{Robustness}

The robustness of homomorphic ciphertext watermarking ensures that, after a sequence of homomorphic operations, the embedded watermark can still be correctly extracted. The homomorphic operations mainly refer to homomorphic addition and homomorphic multiplication. In other words, even after the watermarked ciphertext is evaluated under a polynomial function, the embedded watermark remains correctly recoverable.
We define robustness for homomorphic ciphertext watermarking as follows:
\begin{definition}
     Let $c_{t,1}$ be the ciphertext of \AHE after embedding watermark $w$ into a \AHE ciphertext of $\mathbf{m}_1$ using a watermark scheme, $(c_{t,2},c_{t,3},\cdots)$ are \AHE ciphertexts of $(\mathbf{m}_2,\mathbf{m}_3,\cdots)$.
     Let $f(\cdot)$ be a polynomial function where only additions and multiplications are allowed. Let $c_{t,1}^\prime=f(c_{t,1},c_{t,2},\cdots)$ and $\mathbf{m}_1^\prime=f(\mathbf{m}_1,\mathbf{m}_2,\cdots)$. Let $\sk$ be the secret key of \AHE and  $\exk$ be the extracting key.  
     The watermark scheme has robustness, if $\Extract(\exk,c_{t,1}^\prime)=w$ when $\De(\sk,c_{t,1}^\prime)=\mathbf{m}_1^\prime$.
\end{definition}

In the robustness requirement, we only consider the correctness of watermark extraction under correct decryption. We do not discuss the behavior of watermark extraction when decryption fails. When decryption is incorrect, the decrypted output is typically a random plaintext, which has no meaningful interpretation in most applications; consequently, the extracted watermark in this case is also of limited relevance.

We also require that, after adding small perturbations that do not affect correct decryption, the embedded watermark can still be reliably extracted. In RLWE-based ciphertexts, noise management techniques such as key switching and modulus switching~\cite{acar2018survey} often introduce additional additive noise while preserving correct decryption. A well-designed homomorphic ciphertext watermarking scheme should therefore remain robust against such  noise.



\subsection{Security-Preserving}
The watermarked ciphertext must preserve the security guarantees of the original encryption scheme and should not reveal any information about the plaintext, since perfectly hiding the underlying plaintext is the most fundamental property of a ciphertext. In the setting of homomorphic ciphertext watermarking, the adversary may obtain the embedding key, the embedded watermark, or the watermarking algorithm itself. However, we assume that the adversary cannot obtain the secret key $\sk$ of the underlying HE scheme. We require that watermark embedding does not change the distribution of ciphertexts. In other words, ciphertexts with embedded watermarks should be computationally indistinguishable from ciphertexts without watermarks.
We formalize this property of watermarking as follows:
\begin{definition}
     Assume that $b$ is uniformly sampled from $\{0,1\}$. Let $c_{t,0}$ and $c_{t,1}$ be two ciphertexts of two equal size messages. Let $c_{t,b}$ be the ciphertext after embedding $w$ with a watermarking scheme.  Give $c_{t,b}$ to an adversary. The adversary can ask for the embedded ciphertext of any ciphertext except $c_{t,0}$ and $c_{t,1}$. 
     If every probabilistic polynomial-time adversary can guess \(b\) with probability at most \(1/2+\epsilon(\lambda)\), where \(\epsilon(\lambda)\) is a negligible function, then the watermarking scheme is security-preserving.
\end{definition}

In these properties, robustness captures the preservation of the embedded watermark: after a sequence of operations, the watermark should not be destroyed. In contrast, fidelity concerns the preservation of the hidden plaintext in the ciphertext; that is, the watermark embedding process should not affect the underlying plaintext. Finally, the security requirement ensures that the watermark embedding procedure does not leak any information about the plaintext.

\section{Watermarking Constructions for \AHE}

In this section, we present two proposed watermarking constructions for \AHE ciphertexts: an addition-robust watermarking scheme, \AddRWatermark, and a multiplication-robust watermarking scheme, \MultRWatermark.

\subsection{\AddRWatermark}
\label{sec:add_watermark}

\begin{figure}[!t]
\centering
\includegraphics[width=\linewidth]{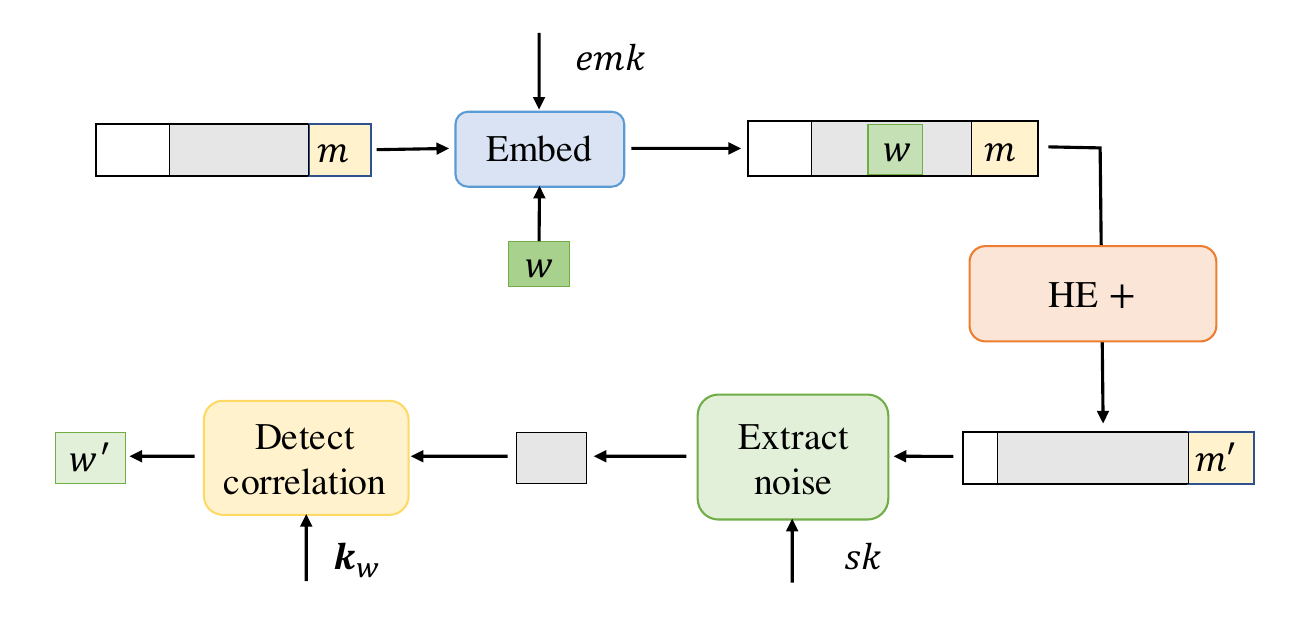}
\caption{Watermarking workflow of \AddRWatermark. The watermark $w$ is embedded into the noise part of ciphertexts. Only homomorphic additions are considered.}
\label{fig:framework1}
\end{figure}

In \AddRWatermark, the watermark is embedded into a single ciphertext. We embed the watermark into the decryption noise of the ciphertext, so that the plaintext hidden in the ciphertext is not affected. During decryption, the decryption noise can be extracted and used for watermark detection. The watermark in this scheme can only be detected by an entity that holds the decryption secret key $\sk$, which prevents unauthorized parties from performing watermark detection. The workflow  is illustrated in \cref{fig:framework1}.

We first describe how to embed a 1-bit watermark into a single ciphertext, and then extend the construction to the multi-ciphertext setting and the embedding of a $t$-bit watermark.

\subsubsection{1-Bit Watermark Embedding}

According to the decryption equation of \AHE, for an entity holding the secret key $\sk$, the noise and the plaintext in a ciphertext $c_t$ can be separated. Specifically, applying modulo-$p$ reduction to $<c_t,\sk>$ completely separates the noise from the plaintext, which is also required for correct decryption in \AHE.

Our main idea is to embed the watermark into the noise component of the ciphertext structure while keeping it separated from the plaintext, thereby preserving correct decryption, as illustrated in \cref{fig:framework1}. 

We assume that the parameters of \AHE have already been generated before watermark embedding, together with the secret key $\sk=(1,\mathbf{s})$.
Let the 1-bit watermark to be embedded be denoted by $w$, where $w$ is mapped to $\{-1,1\}$. We use $w_e$ to denote the mapped watermark value. Specifically, when $w=0$, we set $w_e=-1$, and when $w=1$, we set $w_e=1$. In addition, we use $w_e=0$ to indicate that no watermark is embedded.
We next describe the watermark embedding and detection algorithms in detail.


\textbf{Embedding.}
Let $\chi_w$ be the embedding key distribution over $\mathbf{R}_q$. To prevent brute-force attacks from recovering $\mathbf{k}_w$, $\chi_w$ should provide sufficient randomness. In practice, the encryption noise distribution $\chi_e$ can also be used as the embedding key distribution.
To embed the watermark as part of the ciphertext noise, we first sample $\mathbf{k}_w$ from $\chi_w$ as the embedding key $\emk$, and then determine the embedding intensity $I_w\in \mathbb{Z}^+$.

Similar to spread-spectrum watermarking \cite{cox1997secure}, which embeds watermarks into different frequency components, our scheme embeds the watermark into different noise ranges and uses $I_w$ to shift $\mathbf{k}_w$ across these ranges. A larger embedding intensity reduces the influence of encryption noise on the embedded watermark, thereby improving robustness.

Let the cover ciphertext be $c_t=(\mathbf{c}_1,\mathbf{c}_2)$ with plaintext modulus $p$.
The embedding equation is defined as follows:
\begin{align}
    \label{eq:s_insert_watermark}
    \mathbf{c}^\prime_1=\mathbf{c}_1+pw_eI_w\mathbf{k}_w,
\end{align}
where $p$ is used to ensure that the plaintext remains unaffected.
The watermarked result is $c_t^\prime=(\mathbf{c}^\prime_1,\mathbf{c}_2)$.

We modify only the first component of $c_t$ while keeping the second component unchanged, thereby preserving the decryption structure. The embedding process does not depend on $\sk$, which means that any party can perform watermark embedding on the ciphertext.

\textbf{Detection.}
If the detector knows the original ciphertext $c_t$, the watermark can be trivially extracted from $c_t^\prime-c_t$. However, such detection fails once the ciphertext is modified and therefore cannot provide robustness. We instead consider blind detection based on correlation, where the detector does not know $c_t$.

As illustrated in \cref{fig:framework1}, watermark detection consists of two steps: extracting noise from the ciphertext and detecting correlation.

\textit{Extracting Noise.}
This step requires the secret key $\sk$, and its procedure is similar to decryption.
Let the ciphertext to be detected be $c^\prime_t=(\mathbf{c}^\prime_1,\mathbf{c}^\prime_2)$.
We first divide by $p$ to separate the plaintext component, and then divide by the embedding intensity $I_w$ to extract the watermark from the corresponding noise range.
The extracted noise $\mathbf{e^\prime}$ is computed as follows:
\begin{align}
    \label{eq:extract_noise}
    \mathbf{e^\prime}=\lfloor \frac{\mathbf{c}^\prime_1+\mathbf{c}^\prime_2 \cdot \mathbf{s}}{I_wp} \rceil
\end{align}
where $\lfloor \cdot \rceil$ denotes the coefficient-wise rounding operation.

Since homomorphic operations increase the noise in the ciphertext, the extracted noise may contain both the embedded watermark and accumulated encryption noise. We therefore use correlation-based detection to determine whether the watermark exists.

\textit{Detecting Correlation.}
From a communication perspective, the watermark can be viewed as the transmitted signal, while the encryption noise corresponds to the noise introduced during transmission. According to the construction of \AHE, the encryption noise follows additive Gaussian noise. By the Neyman-Pearson lemma~\cite{neyman1933ix}, linear correlation \cite{miller1999computing} is the optimal method for detecting the presence of $\mathbf{k}_w$. The correlation value can be computed as follows:
\begin{align}
    \label{eq:compute_sm}
    \rho = \frac{1}{N}\times <\mathbf{e^\prime},\mathbf{k}_w>. 
\end{align}

We then determine the embedded information according to a predefined threshold $T_w$ $(T_w>0)$:
\begin{align}
    \label{eq:judge}
    w_e=\left\{
    \begin{array}{cc}
        -1, &\rho \leq -T_w\\
        0, &-T_w<\rho<T_w\\
        1, &\rho \geq T_w
    \end{array}
    \right.
\end{align}

If $w_e=0$, the output indicates that no watermark is embedded in $c_t$. If $w_e=-1$, the detected watermark bit is $0$. If $w_e=1$, the detected watermark bit is $1$.

Since homomorphic addition introduces Gaussian noise, the above detection procedure is robust under homomorphic addition. However, homomorphic addition also increases the variance of the Gaussian noise. Therefore, the threshold $T_w$ depends not only on the embedding intensity $I_w$ and the embedding key distribution $\chi_w$, but also on the expected number of homomorphic additions.






\subsubsection{Multi-Ciphertext Setting}
\label{sec:MCEmbed}

\textbf{Embedding.} 
Assume that there exist $m$ ciphertexts $c_{t,1}, c_{t,2}, \cdots, c_{t,m}$. We first generate a template vector $\mathbf{r}$ by uniformly random sampling, where $\mathbf{r} \in \{-1,1\}^m$.

Next, we compute $\mathbf{r^\prime} = w_e \mathbf{r}$. Each entry $\mathbf{r^\prime}[i]$ is then embedded into the corresponding ciphertext $c_{t,i}$ using \cref{eq:s_insert_watermark}.

\textbf{Detection.}
Given $m$ ciphertexts, we first compute the correlation of each ciphertext using  \cref{eq:compute_sm}. Let the resulting correlation values be $\rho_1, \rho_2, \cdots, \rho_m$.

We then compute the averaged correlation as follows:
\begin{align}
    \label{eq:ref_sm}
    \rho=\frac{1}{m}\sum_i \rho_i \cdot \mathbf{r}[i].
\end{align}

Finally, we apply \cref{eq:judge} to determine whether a watermark is embedded and to recover the embedded information based on the value of $\rho$ and the threshold $T_w$.

Since the overall correlation is obtained by averaging the single-ciphertext correlation values, its variance is reduced by a factor of $\frac{1}{m}$ compared to the single-ciphertext case, leading to a more robust detection procedure.

\subsubsection{t-Bit Embedding}
\label{sec:add_watermark_capcity}
The watermark $w$ may consist of $t$ bits rather than a single bit. A straightforward approach is to embed each bit into a separate ciphertext. However, this becomes impractical when the number of available ciphertexts is limited.

Our detection method relies on linear correlation. When multiple linearly independent embedding keys are used within the same ciphertext, their respective detections do not interfere with each other. Inspired by \cite{kim2017watermarking}, we therefore embed multi-bit watermarks into a single ciphertext to increase the embedding capacity.

Let the binary representation of $w$ be
$w = w_0 + w_1 \cdot 2 + w_2 \cdot 2^2 + \cdots + w_{t-1} \cdot 2^{t-1}$, where $w_i \in \{0,1\}$. Using the single-ciphertext embedding method, each bit $w_i$ is embedded into the same ciphertext $c_t$ with a corresponding embedding key $\mathbf{k}_{w,i}$. The embedding keys are chosen to be linearly independent, i.e., $\langle \mathbf{k}_{w,i}, \mathbf{k}_{w,j} \rangle = 0$ for $i \neq j$. In practice, perfect orthogonality is not required; it is sufficient that their inner products are sufficiently small, which simplifies the construction of embedding keys.

During extraction, each $w_i$ is recovered separately. The full watermark $w$ is correctly recovered only if all bits $w_i$ are successfully extracted.
This construction increases the watermarking capacity from 1 bit to $t$ bits.






\subsection{\MultRWatermark}
\label{sec:mult_watermark}

\begin{figure}[!t]
\centering
\includegraphics[width=\linewidth]{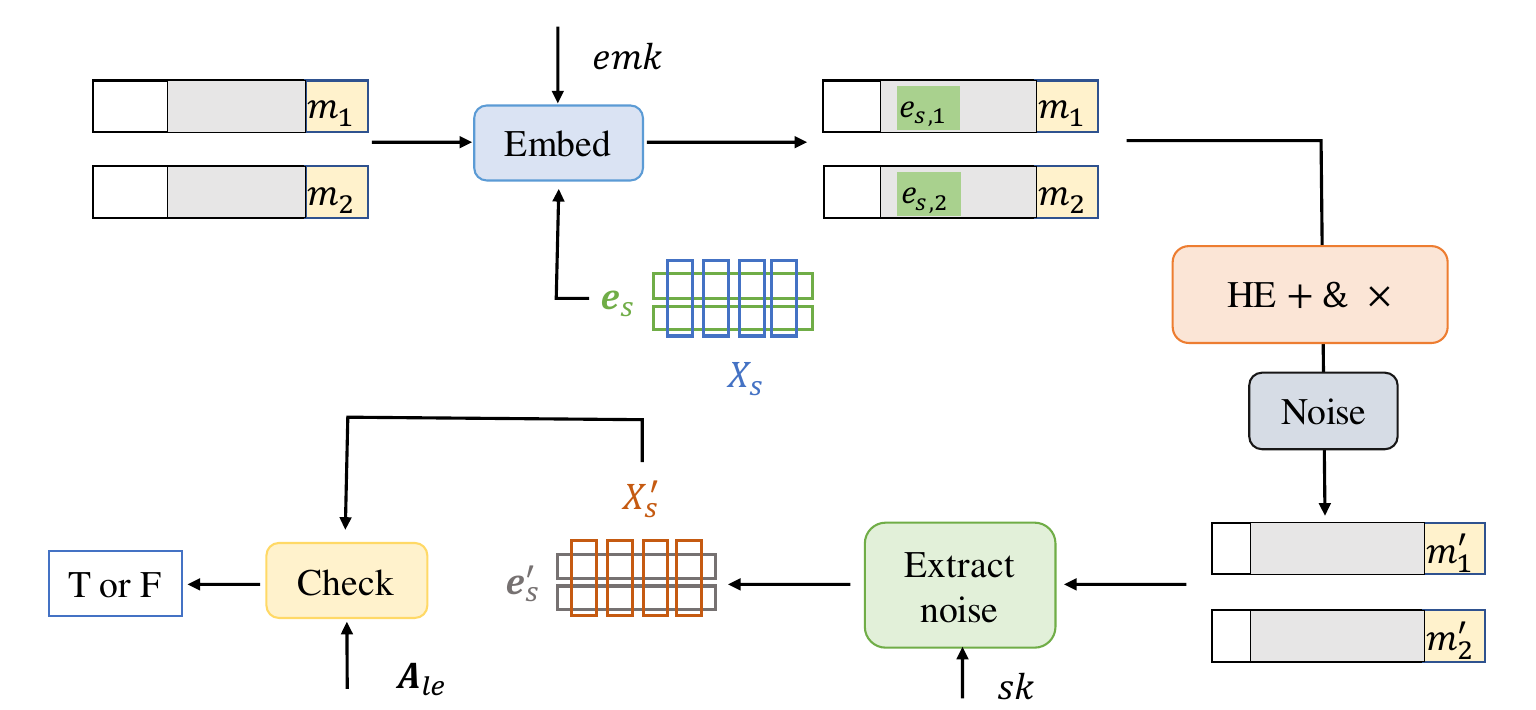}
\caption{Watermarking workflow of \MultRWatermark. The solutions  are rearranged and then embedded into multiple ciphertexts. After extracting, check whether the extracted vectors are the solutions to the system of linear equations. Output T for true or F for false.}
\label{fig:mult_watermark}
\end{figure}

When a multiplication is performed between two ciphertexts, the decryption noises from different ciphertexts are mixed together. This significantly complicates watermark detection, and linear correlation-based detection becomes almost ineffective for recovering the embedded watermark. In this section, we extend the previous embedding scheme to obtain a watermarking method that is robust against ciphertext multiplication.

The key idea is to construct a carefully designed embedding key such that the embedded watermark satisfies an invariant relation under ciphertext multiplication, enabling robust detection after homomorphic multiplication. It is worth noting that the proposed scheme in this subsection is a 0-bit watermarking scheme, meaning that it only detects the presence of a watermark without extracting any specific message.

Ciphertext multiplication introduces substantial noise, and the noise in \pkscheme is typically much larger than that in \skscheme. Therefore, we restrict our attention to embedding on \skscheme ciphertexts, where the noise level is relatively small. The watermark is carried by a set of ciphertexts, which are treated as a single logical unit, and the same operation is applied to each ciphertext. Such a scenario is common in practice, for example in multimedia processing, where identical batch operations are applied to different images.
We now present the embedding and detection procedures.




\subsubsection{Embedding}
By analyzing multiplication over $\mathbf{R}_q$, we observe that the coefficients of the product polynomial are linear combinations of coefficients from different terms of the original polynomials. Moreover, any linear combination of solutions to a homogeneous linear system remains a solution to the same system \cite{leon2006linear}.
Based on this observation, we embed the solution space of a homogeneous linear system into a set of ciphertexts, thereby obtaining a watermark that remains invariant under the linear combinations induced by ciphertext multiplication.


Assume that the homogeneous linear system $\mathbf{A}_{le}\cdot X_{le}=0$ has infinitely many solutions, where $X_{le}$ is a vector consisting of $m$ unknown variables. The matrix $\mathbf{A}_{le}$ has fewer than $m$ rows and exactly $m$ columns. Let $X_{s,1}, X_{s,2}, \cdots, X_{s,N}$ denote $N$ solutions of this system with relatively small norms, where different solutions are allowed to be identical but cannot be the zero vector. Note that each solution is a vector of length $m$. The embedding key $\emk$ is the matrix $\mathbf{A}_{le}$.


Next, we construct $m$ noise vectors $\mathbf{e}_{s,1}, \mathbf{e}_{s,2}, \cdots, \mathbf{e}_{s,m}$ as follows:
\begin{align}
    \label{eq:le_es_construct}
    \mathbf{e}_{s,i}[j]=X_{s,j}[i]
\end{align}
where $i=1,2,\cdots,m$ and $j=1,2,\cdots,N$.

We then embed these $m$ noise vectors into the corresponding $m$ ciphertexts $\mathbf{c}_{t,1}, \mathbf{c}_{t,2}, \cdots, \mathbf{c}_{t,m}$:
\begin{equation}
\label{eq:LE_insert_watermark}
    \begin{aligned}
        \mathbf{c}_{r,i}[0]=&\mathbf{c}_{t,i}[0]+pI_w\mathbf{e}_{s,i},  \\
        \mathbf{c}_{r,i}[1]=&\mathbf{c}_{t,i}[1]
    \end{aligned}
\end{equation}
where $i=1,2,\cdots,m$.




\subsubsection{Detection}
Detection in this scheme mainly consists of two steps. The first step extracts the embedded noise, and the second step determines whether the vectors reconstructed from the extracted noise are nonzero solutions of the original linear system.

We first use \cref{eq:extract_noise} to extract the noise from each ciphertext. The extracted noise is then transformed, following the inverse process of \cref{eq:le_es_construct}, to recover $N$ vectors. If all recovered vectors are zero vectors, we conclude that no watermark is embedded. If all recovered vectors are nonzero solutions to $\mathbf{A}_{le}\cdot X_{le}=0$, then the ciphertexts are determined to contain a watermark.

Since the extraction process requires the secret key $\sk$, watermark detection can only be performed by the holder of the secret key.



\subsubsection{Construction of the Embedding Key}
\label{sec:A_le_construct}
After adding random noise to a solution of a linear system, the resulting vector may no longer satisfy the original system. Therefore, during embedding, the watermark must be shifted outside the range of the intrinsic decryption noise in practical ciphertexts to enable reliable extraction. This requires the embedding intensity to be sufficiently large. However, the tolerable decryption noise budget is limited, so the additional noise introduced by watermark embedding should be kept as small as possible. Consequently, we require the linear system to admit integer solutions with small infinity norms.

Constructing a linear system that satisfies these requirements is not straightforward. We therefore construct the linear system in a reverse manner.

Assume that the target matrix $\mathbf{A}_{le}$ is a $k \times m$ matrix with $k<m$. We first randomly sample $m$ small integers from the range $[-B_w,B_w]$ to construct a base solution $X_1$. By construction, $X_1$ has a sufficiently small norm to satisfy the watermark embedding requirement.

Next, we treat each entry of $\mathbf{A}_{le}$ as a variable and construct the equation system
$\mathbf{A}_{le} \cdot X_1=0$.
This system contains $k$ equations but has $km$ variables, and therefore admits infinitely many solutions.

Since $\mathbf{A}_{le}$ is not required to be an integer matrix, it is easy to obtain a matrix $\mathbf{A}_{le}$ satisfying the above condition.

Clearly, $\mathbf{A}_{le} \cdot X_{le}=0$ has the solution $X_1$. Moreover, because $k<m$, the constructed linear system admits infinitely many solutions.




\subsubsection{Embedding Capacity}

\MultRWatermark does not support repeated embedding within the same set of ciphertexts. Therefore, the capacity extension method described in \cref{sec:add_watermark_capcity} cannot be directly applied. To address this limitation, we present an alternative approach for extending the watermarking capacity.

We first extend the 0-bit watermark to a 1-bit watermark, and then further generalize it to a $t$-bit watermark. We construct two linear systems with coefficient matrices $\mathbf{A}_{le,0}$ and $\mathbf{A}_{le,1}$, respectively. The solution spaces of these two systems are designed such that their nonzero solutions do not overlap. The solution space of $\mathbf{A}_{le,0}\cdot X_{le}=0$ is used to represent the watermark value $0$, while the solution space of $\mathbf{A}_{le,1}\cdot X_{le}=0$ represents the watermark value $1$. When the embedded watermark is $0$, we embed a vector from the solution space of $\mathbf{A}_{le,0}\cdot X_{le}=0$. When the embedded watermark is $1$, we embed a vector from the solution space of $\mathbf{A}_{le,1}\cdot X_{le}=0$. During detection, we determine whether the extracted vector is a solution of either of the two linear systems. This construction yields a 1-bit watermarking scheme.

Assume that a large number of ciphertexts are available. We select $t \times m$ ciphertexts and partition them into $t$ groups, each consisting of $m$ ciphertexts. Each group is used to embed one bit of information, thereby extending the capacity to $t$ bits. All $t$ groups use the same linear system construction; thus, it suffices to generate two linear systems whose nonzero solution spaces do not overlap.

\section{Theoretical Analysis}

We use $\E(\cdot)$ and $\var(\cdot)$ to denote expectation and variance, respectively. And $\cov(\cdot,\cdot)$ denotes the covariance.
Let $\chi_u$ be the uniform distribution over $\mathbf{R}_q$.
Let $\chi_s$ be a distribution over $\mathbf{R}_q$ with coefficients independently sampled from $\{-1,0,1\}$.
Let $\chi_e$ be a bounded discrete Gaussian distribution ~\cite{albrecht2021homomorphic} over $\mathbf{R}_q$, with mean $0$ and variance $\sigma^2$. We assume samples from $\chi_e$ are bounded by $B_e$. For simplicity, we set $\chi_w=\chi_e$.
We write $x \sim \chi$ to denote that $x$ follows distribution $\chi$. Due to space constraints, all lemmas are deferred to Appendix~\ref{appendix:lemmas}.

The secret key is $\sk=(1,\mathbf{s})$, and the public key is $\pk=(\mathbf{k}_1,\mathbf{k}_2)=(\mathbf{a}_k \cdot \mathbf{s}+p\mathbf{e}_k,-\mathbf{a}_k)$, where $\mathbf{a}_k \sim \chi_u$ and $\mathbf{e}_k \sim \chi_e$.
Let $w_e \in \{-1,0,1\}$ denote the embedded message for a single ciphertext, where $0$ indicates no embedding. In \AddRWatermark, the embedding key is sampled as $\mathbf{k}_w \sim \chi_e$. 





\subsection{Analysis of \AddRWatermark}
In this subsection, we mainly analyze the correctness of \AddRWatermark. 
The security-preserving, fidelity, and robustness analyses are provided in Appendix~\ref{appendix:other_analysis_add}.

\subsubsection{Correctness}
The detection of watermark information is not a continuous function, but rather a threshold-based decision. Therefore, we analyze the expectation and variance of the extracted correlation value under different embedding watermarks.

According to \cref{eq:extract_noise}, the noise term is divided by $pI_w$ and then rounded. Since $\Vert \mathbf{m} \Vert \leq p/2$, we have $\lfloor \frac{\mathbf{m}}{pI_w} \rceil= 0$. Therefore, in the following analysis, we ignore the effect of $\mathbf{m}$ on the extraction noise.
In addition, we neglect the rounding error introduced during noise extraction.

\begin{theorem}
\label{th:sk_e_v}
Let a ciphertext of \skscheme be
$c_t = (\mathbf{a} \cdot \mathbf{s} + p\mathbf{e}_1 + \mathbf{m}, -\mathbf{a})$,
where $\mathbf{a} \sim \chi_u$ and $\mathbf{e}_1 \sim \chi_e$.
After embedding the watermark $w_e$ using \AddRWatermark, the expectation of the correlation value $\rho_{sk}$ is $w_e\sigma^2$, and its variance is $\frac{\sigma^4}{NI_w^2} + \frac{2w_e^2\sigma^4}{N}$.
\end{theorem}



\begin{proof}
By \cref{eq:noise_in_sk} and \cref{eq:extract_noise}, the extracted noise from ciphertext $c_t$ is $\frac{\mathbf{e}_1}{I_w}+w_e\mathbf{k}_w$.

\begin{equation*}
    \begin{aligned}
         \rho_{sk} &=\frac{1}{I_wN}\cdot \langle \mathbf{e}_1+w_eI_w\mathbf{k}_w,\mathbf{k}_w\rangle\\
    &=\frac{1}{NI_w} (\langle\mathbf{e}_1,\mathbf{k}_w\rangle+w_eI_w\langle\mathbf{k}_w,\mathbf{k}_w\rangle)
    \end{aligned}
\end{equation*}

By \cref{lm:xy_e}, we have $\E(\langle \mathbf{e}_1,\mathbf{k}_w\rangle)=0$ and $\var(\langle\mathbf{e}_1,\mathbf{k}_w\rangle)=N\sigma^4$.
By \cref{lm:xx_e}, we have $\E(\langle \mathbf{k}_w,\mathbf{k}_w\rangle)=N\sigma^2$ and $\var(\langle\mathbf{k}_w,\mathbf{k}_w\rangle)=2N\sigma^4$.
Thus, we have
\begin{equation*}
    \begin{aligned}
        \E(\rho_{sk})&=\E(\frac{1}{NI_w} (<\mathbf{e}_1,\mathbf{k}_w>+w_eI_w<\mathbf{k}_w,\mathbf{k}_w>))\\
        &=\frac{1}{NI_w}\E(<\mathbf{e}_1,\mathbf{k}_w>)+\frac{w_e}{N}\E(<\mathbf{k}_w,\mathbf{k}_w>)\\
        &=w_e\sigma^2.
    \end{aligned}
\end{equation*}

By \cref{lm:cov}, we have $\cov(\langle \mathbf{e}_1,\mathbf{k}_w\rangle,\langle\mathbf{k}_w,\mathbf{k}_w\rangle)=0$. Therefore, we have
\begin{equation*}
    \begin{aligned}
        \var(\rho_{sk})&=\var(\frac{1}{NI_w} (\langle\mathbf{e}_1,\mathbf{k}_w\rangle+w_eI_w\langle\mathbf{k}_w,\mathbf{k}_w\rangle))\\
        &=\frac{1}{N^2I_w^2}\var(\langle\mathbf{e}_1,\mathbf{k}_w\rangle)+\frac{w_e^2}{N^2}\var(\langle\mathbf{k}_w,\mathbf{k}_w\rangle)\\
        &=\frac{\sigma^4}{NI_w^2}+\frac{2w_e^2\sigma^4}{N}.       
    \end{aligned}
\end{equation*}

\end{proof}

    


    

Next, we estimate the extraction correctness probability based on the expectation and variance.

We first analyze the case where the detection result is $1$ for a watermark embedded in a \skscheme ciphertext.
We use $P_r\{\cdot\}$ to denote the probability that an event occurs.
Since $\E(\rho_{sk})=\sigma^2$ and $\var(\rho_{sk})=\frac{\sigma^4}{NI_w^2}+\frac{2\sigma^4}{N}$ when $w_e=1$, we assume that the decision threshold satisfies $T_w<\sigma^2$.

When embedding $w_e=1$, by Chebyshev's Inequality \cite{cohen2015markov}, we have
\begin{equation*}
    \begin{aligned}
        P_r\{\rho_{sk}< T_w\}&<P_r\{\vert \rho_{sk}-\sigma^2 \vert>\sigma^2-T_w \}\\
    &<\frac{1}{(\sigma^2-T_w)^2}\cdot (\frac{\sigma^4}{NI_w^2}+\frac{2\sigma^4}{N}).
    \end{aligned}
\end{equation*}
That is, the True Positive Rate (TPR) is at least
\begin{equation*}
    1-\frac{1}{(\sigma^2-T_w)^2}\cdot (\frac{\sigma^4}{NI_w^2}+\frac{2\sigma^4}{N}).    
\end{equation*}

When $w_e=0$, we have $\E(\rho_{sk})=0$ and $\var(\rho_{sk})=\frac{\sigma^4}{NI_w^2}$.
Therefore,
\begin{equation*}
    \begin{aligned}
        P_r\{\rho_{sk}>= T_w\}&<P_r\{\vert \rho_{sk} \vert>T_w \}<\frac{\sigma^4}{T_w^2NI_w^2}.
    \end{aligned}
\end{equation*}

When $w_e=-1$, we have $\E(\rho_{sk})=-\sigma^2$ and $\var(\rho_{sk})=\frac{\sigma^4}{NI_w^2}+\frac{2\sigma^4}{N}$.
Thus,
\begin{equation*}
    \begin{aligned}
        P_r\{\rho_{sk}>= T_w\}&<P_r\{\vert \rho_{sk}+\sigma^2 \vert>T_w+\sigma^2 \}\\
    &<\frac{1}{(\sigma^2+T_w)^2}\cdot (\frac{\sigma^4}{NI_w^2}+\frac{2\sigma^4}{N}).
    \end{aligned}
\end{equation*}




Therefore, the False Positive Rate (FPR) is upper bounded by
\begin{equation*}
\frac{\sigma^4}{T_w^2NI_w^2} + \frac{1}{(\sigma^2+T_w)^2}\cdot \left(\frac{\sigma^4}{NI_w^2}+\frac{2\sigma^4}{N}\right).
\end{equation*}


We aim for a high TPR and a low FPR. Both objectives can be achieved by increasing $N$ and $I_w$.
When the detection result is $-1$, a similar derivation leads to the same conclusion.
When the embedding carrier is a \pkscheme ciphertext, the expectation of $\rho$ is the same as in \skscheme, which leads to analogous results.

Typically, $\sigma$ is treated as a fixed parameter, commonly set to $3.2$ \cite{albrecht2021homomorphic}.
Under this setting, the expectation of $\rho$ is fixed. Therefore, both TPR and FPR are determined by the variance of $\rho$. A smaller variance implies a higher TPR and a lower FPR.
By adjusting the parameters $N$ and $I_w$, we can control the variance of $\rho$, thereby increasing the TPR and decreasing the FPR.
This allows the different values of $w_e$ to be distinguished with high probability.
In particular, larger values of $N$ and $I_w$ lead to a smaller variance of $\rho$, resulting in a higher TPR, a lower FPR, and a clearer separation between different values of $w_e$.




\subsection{Analysis of \MultRWatermark}

Due to the space limit, we only analyze the correctness of \MultRWatermark in this subsection. The analyses of other properties are provided in Appendix~\ref{appendix:other_analysis_mult}.


\subsubsection{Correctness}
\begin{theorem}
Assume that there are $m$ fresh ciphertexts $c_{t,1},\cdots,c_{t,m}$.
When the watermark embedding intensity satisfies $I_w>2B_e+1$, the watermark embedded into these $m$ ciphertexts using \MultRWatermark can be correctly extracted.
\end{theorem}
    
\begin{proof}
Let $c_{t,i}=(\mathbf{a}_i \cdot \mathbf{s}+p\mathbf{e}_i+\mathbf{m}_i,-\mathbf{a}_i)$.
Let $\mathbf{e}_{s,i}$ denote the $i$-th noise term constructed from the solution $X_{s,1},X_{s,2},\cdots,X_{s,N}$.
The ciphertext after watermark embedding is
\begin{equation*}
c_{t,i}^\prime=(\mathbf{a}_i \cdot \mathbf{s}+p\mathbf{e}_i+\mathbf{m}_i+pI_w\mathbf{e}_{s,i},-\mathbf{a}_i).
\end{equation*}

We now prove that when $I_w>2B_e+1$, the noise extracted from $c_{t,i}^\prime$ according to \cref{eq:extract_noise} satisfies $\mathbf{e}_{s,i}^\prime=\mathbf{e}_{s,i}$.

The inner product between the secret key $\sk$ and $c_{t,i}^\prime$ is
\begin{equation*}
    \begin{aligned}
    \langle\sk,c_{t,i}^\prime\rangle&=\mathbf{a}_i \cdot \mathbf{s}+p\mathbf{e}_i+\mathbf{m}_i+pI_w\mathbf{e}_{s,i}+(-\mathbf{a}_i\cdot \mathbf{s})\\
    &=p\mathbf{e}_i+\mathbf{m}_i+pI_w\mathbf{e}_{s,i}.
    \end{aligned}
\end{equation*}

Since $\Vert \mathbf{e}_i \Vert \leq B_e$ and $\Vert \mathbf{m}_i \Vert <\frac{p}{2}$, when $I_w>2B_e+1$, we have
\begin{equation*}
\frac{1}{I_wp}(p\mathbf{e}_i+\mathbf{m}_i)<0.5.
\end{equation*}
Therefore,
\begin{equation*}
\mathbf{e}_{s,i}^\prime=\lfloor\frac{p\mathbf{e}_i+\mathbf{m}_i+pI_w\mathbf{e}_{s,i}}{I_wp}  \rceil=\mathbf{e}_{s,i}.
\end{equation*}





\end{proof}

\section{Application Scenario}

In this section, we describe a
representative application in privacy-preserving federated learning.
In federated learning with HE, clients encrypt their
local model updates and send the ciphertexts to an aggregation server.
The server performs aggregation directly over ciphertexts and returns the
aggregated ciphertext to the model owner. Since intermediate ciphertexts
remain encrypted during computation, a client cannot directly prove that
its encrypted update has contributed to the final aggregation result.
Meanwhile, the model owner may lack an effective mechanism to verify
whether a specific encrypted update has been used during aggregation.

Ciphertext watermarking provides a lightweight approach for contribution
verification. Before uploading an encrypted update, a client embeds a
watermark into its ciphertext. The watermark is transparent to the
aggregation server and does not affect the homomorphic evaluation
procedure. After receiving the aggregated ciphertext, the model owner
decrypts the result and extracts the watermark to verify the presence of
the client's encrypted update in the aggregation result. This provides
evidence of ciphertext contribution without revealing the underlying
plaintext data.

For standard federated averaging, where encrypted updates are aggregated
using homomorphic addition, the proposed \AddRWatermark scheme enables
watermark verification after aggregation. For weighted aggregation or
scenarios involving homomorphic multiplication, the proposed
\MultRWatermark scheme further supports verification under multiplicative
operations.

\section{Experimental Results}

We implemented the proposed watermark embedding schemes in Python. For the security of RLWE, we set the parameters according to the 128-bit security recommendations in \cite{albrecht2021homomorphic}. By default, we set $N=2048$, $p=65537$, and $\sigma=3.2$, with $q$ being a 55-bit prime number.
All experiments were conducted on a machine running Ubuntu 20.04.6 LTS Linux, equipped with an Intel Xeon Gold 6145 CPU (2.0 GHz) and 502 GB RAM.

\subsection{Evaluation for \AddRWatermark}
\subsubsection{Correctness Verification}
To ensure reliable watermark extraction, the distribution of the correlation value $\rho$ should be concentrated around its expectation. In addition, the distributions corresponding to different watermark values should be separated as much as possible.

\begin{figure}[!t]
	\centering  
	\subfigbottomskip=1pt 
	\subfigcapskip=1pt 
    \subfigure[$\rho_{sk}$ for \skscheme ciphertexts]{
    \label{fig:distribute_sk}
    \includegraphics[width=0.45\linewidth]{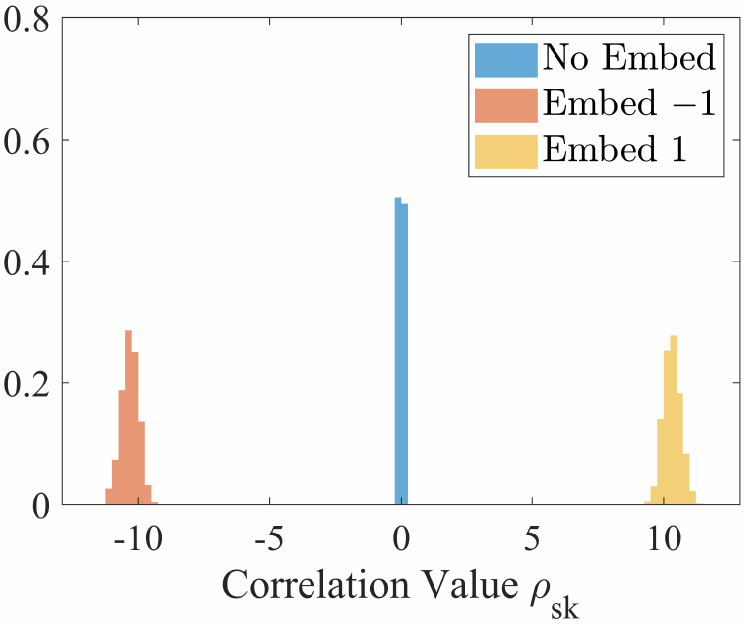}}
	\subfigure[$\rho_{pk}$ for \pkscheme ciphertexts]{
    \label{fig:distribute_pk}
    \includegraphics[width=0.45\linewidth]{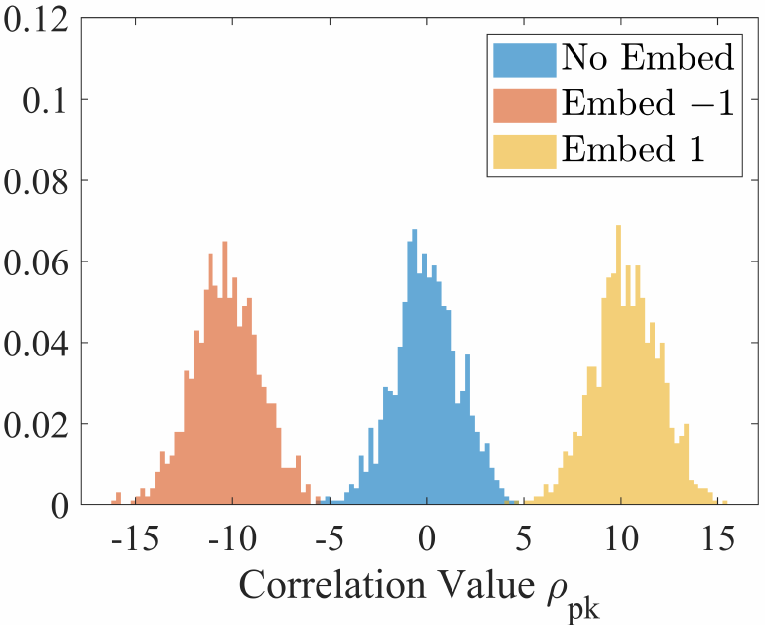}}
	\caption{
 Normalized frequency distribution histogram of $\rho$ with embedding intensity $I_w=7$.
}
\label{fig:distribute_i_7}
\end{figure}

In \cref{fig:distribute_i_7}, we present the histogram of $\rho$ frequencies to approximate its probability density distribution.  
For each scenario, we repeated the embedding process $1000$ times, then extracted the noise and computed the correlation value $\rho$. As shown in the figure, $\rho$ is concentrated around its expectation, consistent with our anticipated behavior, and its distribution resembles a normal distribution.

When $I_w=7$, according to \cref{fig:distribute_sk}, watermarks embedded in \skscheme ciphertexts can be clearly distinguished. However, for watermarks embedded in \pkscheme ciphertexts, the correlation values $\rho_{pk}$ partially overlap around $-5$ and $5$, as illustrated in \cref{fig:distribute_pk}.



\begin{figure}[!t]
	\centering  
	\subfigbottomskip=1pt 
	\subfigcapskip=1pt 
    \subfigure[$\rho_{sk}$ of \skscheme ciphertexts]{
    \label{fig:distribute_sk_i_100}
    \includegraphics[width=0.45\linewidth]{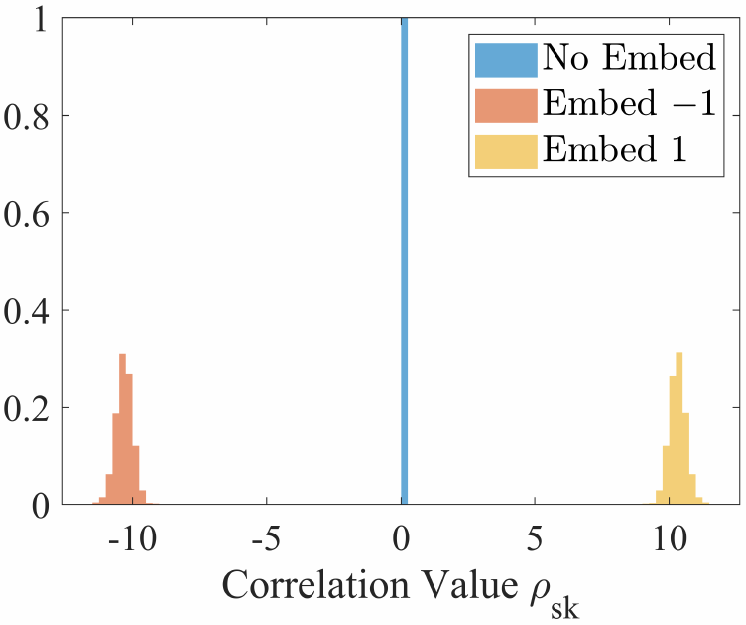}}
	\subfigure[$\rho_{pk}$ of \pkscheme ciphertexts]{
    \label{fig:distribute_pk_i_100}
    \includegraphics[width=0.45\linewidth]{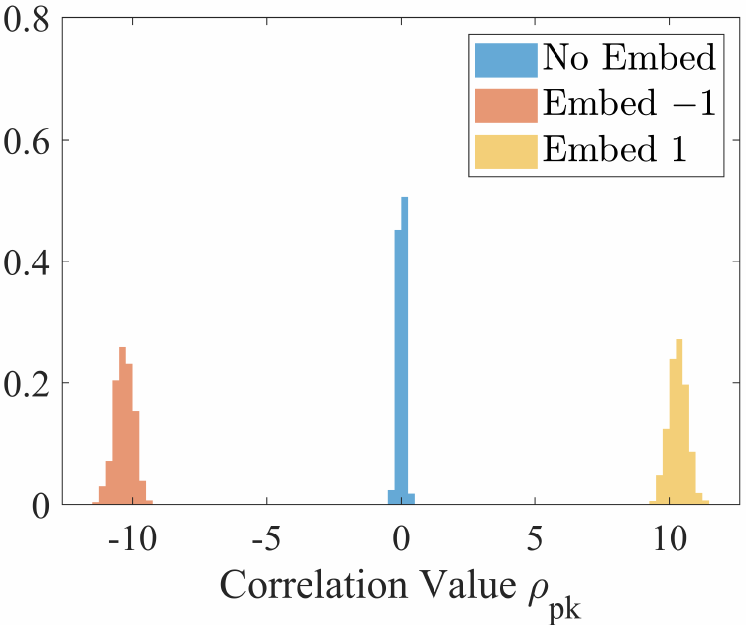}}
	\caption{
 Normalized frequency distribution histogram of $\rho$ with embedding intensity $I_w=100$.
}
\label{fig:distribute_i_100}
\end{figure}

As $I_w$ increases, the expectation of $\rho$ remains unchanged for each embedding case, while the variance decreases, resulting in a more concentrated distribution. As shown in \cref{fig:distribute_i_100}, a narrower distribution of $\rho$ facilitates the distinction between different embedding cases. In particular, as illustrated in \cref{fig:distribute_pk_i_100}, the correlation values $\rho_{pk}$ computed from \pkscheme ciphertexts no longer overlap.
According to \cref{fig:distribute_i_7,fig:distribute_i_100}, under the same parameter settings, the variance of $\rho_{sk}$ is smaller than that of $\rho_{pk}$. This observation is consistent with our theoretical analysis.


\begin{figure}[!t]
	\centering  
	\subfigbottomskip=1pt 
	\subfigcapskip=1pt 
    \subfigure[$\rho_{sk}$ of \skscheme ciphertexts]{
    \label{fig:intensity_sk}
    \includegraphics[width=0.45\linewidth]{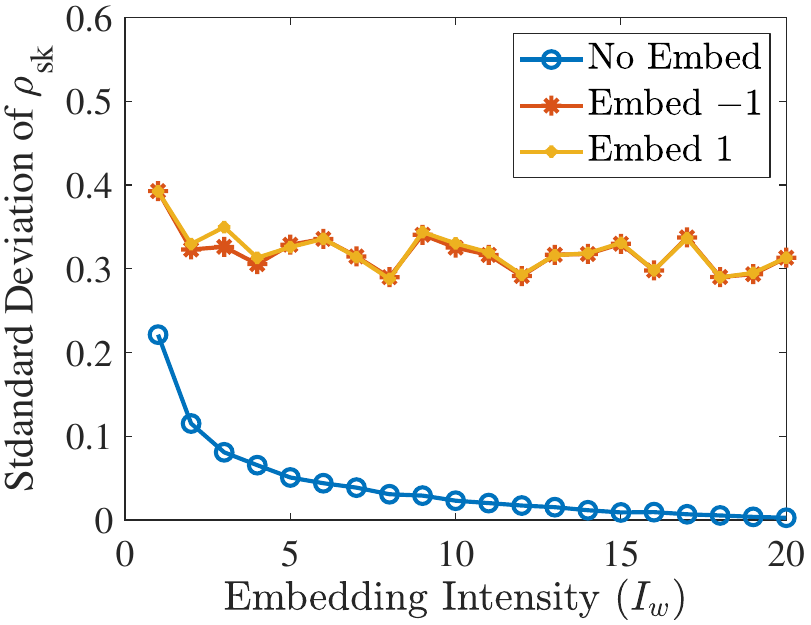}}
	\subfigure[$\rho_{pk}$ of \pkscheme ciphertexts]{
    \label{fig:intensity_pk}
    \includegraphics[width=0.45\linewidth]{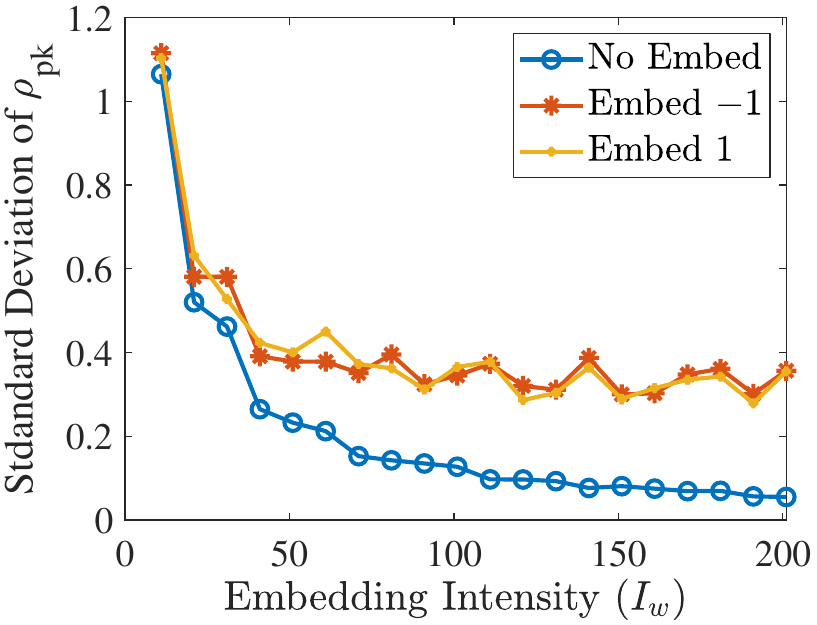}}
	\caption{
 Standard deviation of $\rho$ for different embedding intensities.
}
\label{fig:intensity}
\end{figure}

\subsubsection{Impact of Different Embedding Intensities}
For each $I_w$, we repeated the experiment $20$ times and computed the standard deviation to evaluate the impact of $I_w$ on the correlation value $\rho$, where a smaller standard deviation indicates a more concentrated distribution, which implies lower extraction noise and less interference during watermark extraction. 
We present the experimental results in \cref{fig:intensity}. As $I_w$ increases, the standard deviation of $\rho$ gradually decreases. In addition, the effect of $I_w$ is most pronounced on the extracted $\rho$ in the absence of any embedded watermark.
As shown in \cref{fig:intensity_sk}, for \skscheme ciphertexts, $I_w=7$ provides a reasonable trade-off. At this point, the standard deviation of $\rho_{sk}$ is sufficiently small. Since we need to ensure correct decryption, we do not want $I_w$ to be too large. For \pkscheme ciphertexts, $I_w = 100$ provides a better trade-off, as illustrated in \cref{fig:intensity_pk}.

\begin{figure}[!t]
	\centering  
	\subfigbottomskip=1pt 
	\subfigcapskip=1pt 
    \subfigure[$\rho_{sk}$ of \skscheme]{
    \label{fig:attack1_sk}
    \includegraphics[width=0.45\linewidth]{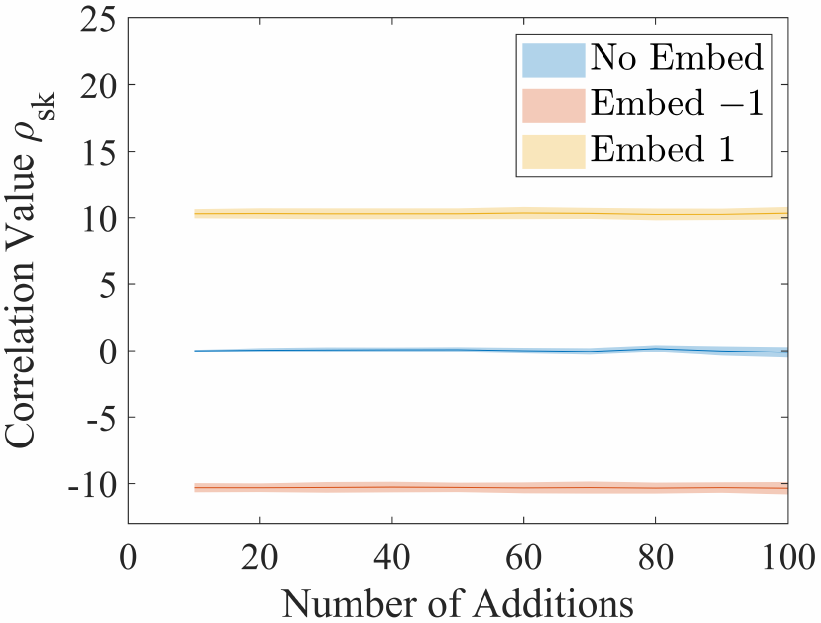}}
	\subfigure[$\rho_{pk}$ of \pkscheme]{
    \label{fig:attack1_pk}
    \includegraphics[width=0.45\linewidth]{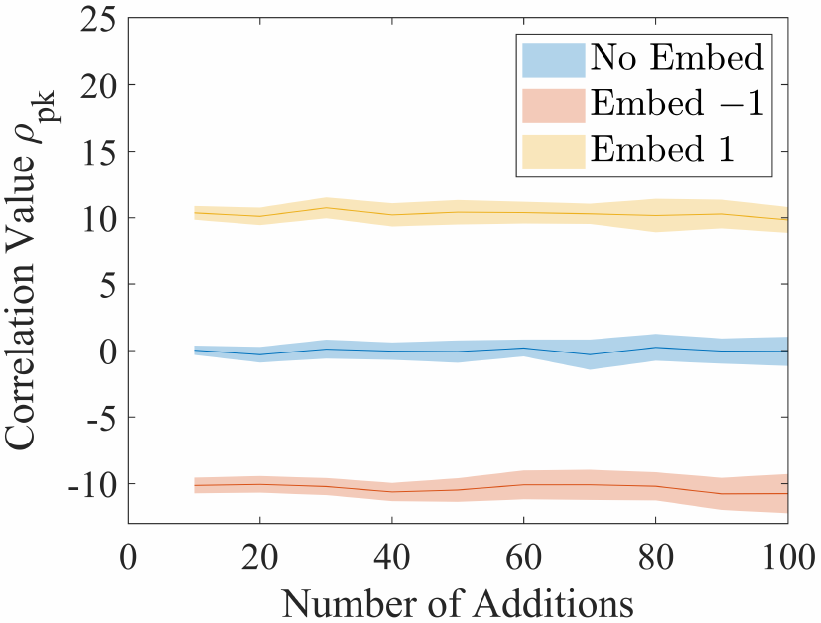}}
    \caption{Comparison of $\rho$ under different embedding values after adding non-watermarked ciphertexts.}

\label{fig:attack1}
\end{figure}

\begin{figure}[!t]
	\centering  
	\subfigbottomskip=1pt 
	\subfigcapskip=1pt 
    \subfigure[$\rho_{sk}$ of \skscheme]{
    \label{fig:attack3_sk}
    \includegraphics[width=0.45\linewidth]{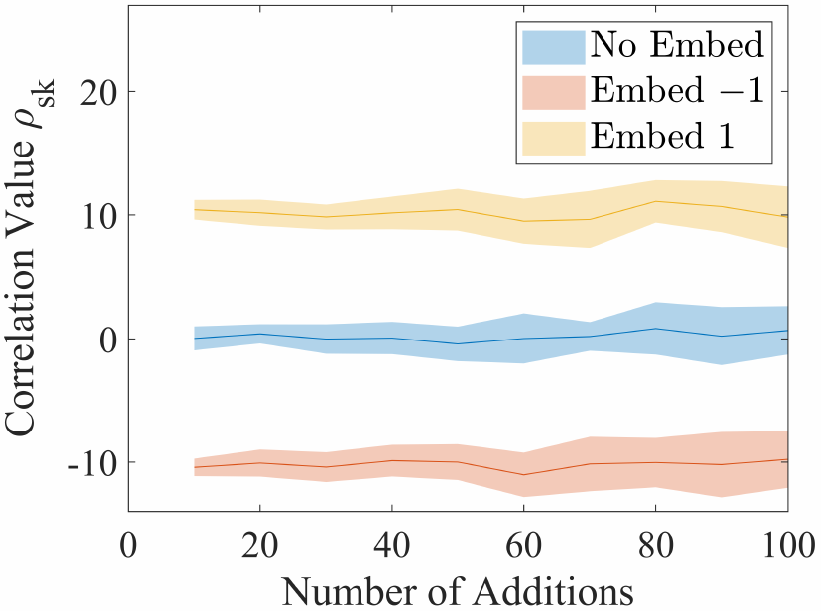}}
	\subfigure[$\rho_{pk}$ of \pkscheme]{
    \label{fig:attack3_pk}
    \includegraphics[width=0.45\linewidth]{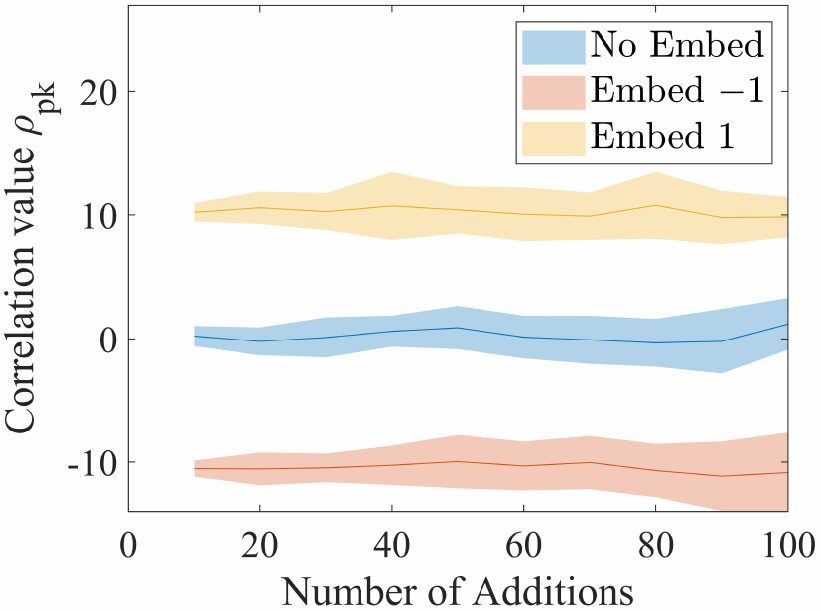}}
	\caption{
 Comparison of $\rho$ under different embedding values after adding other ciphertexts with different embedding keys.
}
\label{fig:attack3}
\end{figure}

\subsubsection{Robustness of the Watermark}

We repeat each experiment $20$ times independently. The mean $\pm$ standard deviation is used as the error band to improve the reliability of the evaluation. For \skscheme ciphertexts, we set $I_w = 7$, while for \pkscheme ciphertexts we set $I_w = 100$.

In \cref{fig:attack1}, we present the value of $\rho$ when multiple ciphertexts without embedded watermarks are added. The expectation of $\rho$ remains nearly unchanged. As the number of additions increases, the variance grows, which is particularly evident in both unwatermarked ciphertexts and \pkscheme ciphertexts, as shown in \cref{fig:attack1_pk}.

In \cref{fig:attack3}, we present the results obtained when adding ciphertexts containing watermarks generated under different embedding keys. Compared with the results in \cref{fig:attack1}, adding a single watermarked ciphertext leads to a more pronounced increase in the standard deviation of $\rho$. However, different embedding cases can still be clearly distinguished.

Overall, our watermark is robust against the noise introduced by homomorphic additions. Although the standard deviation of the extracted $\rho$ increases, there remains sufficient separation to ensure correct watermark extraction.




\subsubsection{Resistance to Random Noise Attack}

Since the goal of a watermark attacker is to remove the watermark without affecting the plaintext, we add noise that does not affect decryption, i.e., multiples of $p$. Let $B_a$ denote the norm of the attack noise. Specifically, the coefficients of the noise $\mathbf{e}$ are uniformly sampled from the interval $[-B_a, B_a]$. Given a ciphertext $c_t=(\mathbf{c}_1, \mathbf{c}_2)$, the attacked ciphertext is $c_t^\prime=(\mathbf{c}_1 + p\mathbf{e}, \mathbf{c}_2)$. 
We set the detection threshold to $T_w = 5$ and repeat the experiment $1000$ times for each $B_a$ to calculate the attack success rate. A watermark of $0$ or $1$ is randomly embedded for each trial; if the extracted watermark does not match the embedded one, the attack is considered successful.

The results in \cref{fig:add_noise} show that when the noise level is small, the impact on watermark extraction accuracy is negligible, indicating that \AddRWatermark is robust to random noise to a certain extent. When $B_a < 100$, the addition of random noise has almost no effect on watermark extraction. As $B_a$ increases, the attack success rate gradually increases. However, larger $I_w$ can be used to counteract stronger noise. Under the same noise level, a higher $I_w$ leads to a lower attack success rate. Although \AddRWatermark requires a larger $I_w$ when applied to \pkscheme ciphertexts, the same overall trend is observed for both \pkscheme and \skscheme ciphertexts as shown in \cref{fig:add_noise_pk}.

\begin{figure}[!t]
	\centering  
	\subfigbottomskip=1pt 
	\subfigcapskip=1pt 
    \subfigure[\skscheme ciphertexts]{
    \label{fig:add_noise_sk}
    \includegraphics[width=0.45\linewidth]{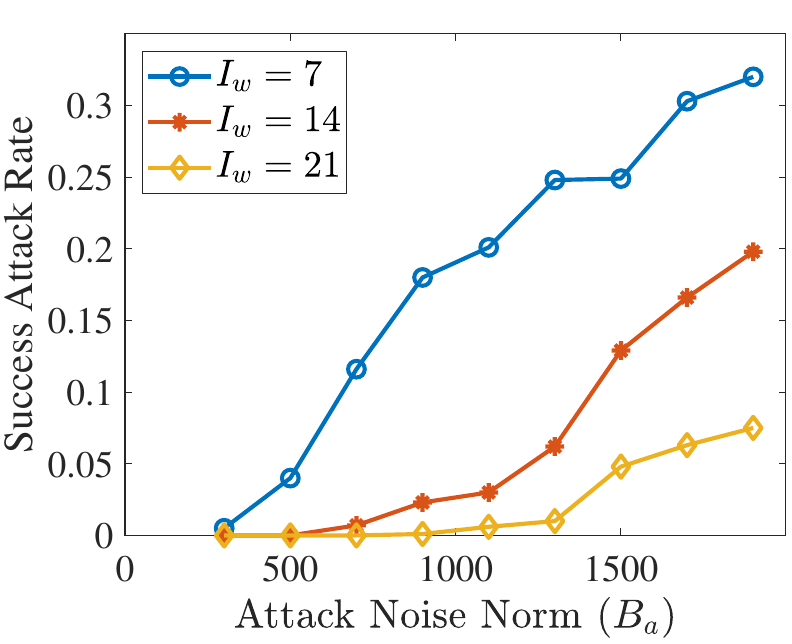}}
	\subfigure[\pkscheme ciphertexts]{
    \label{fig:add_noise_pk}
    \includegraphics[width=0.45\linewidth]{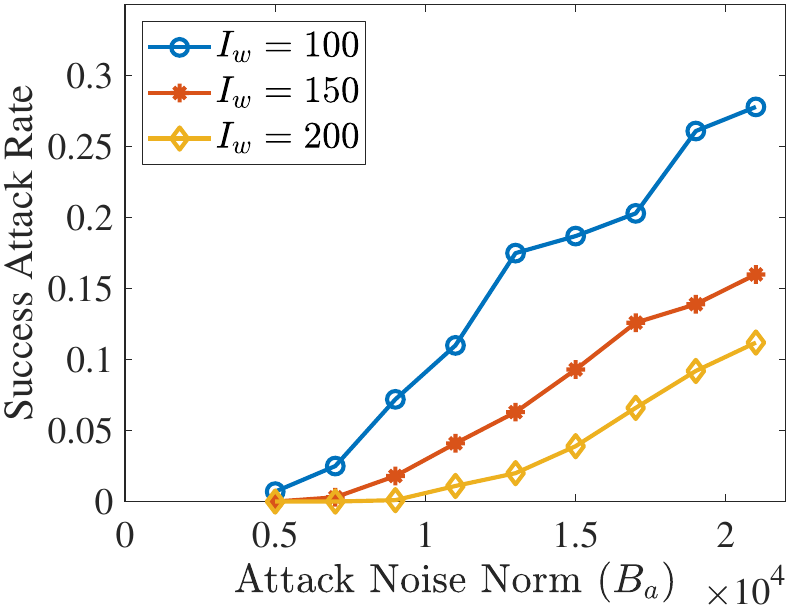}}
	\caption{
 Success attack rate for adding uniformly random noise. We set $T_w=5$.
}
\label{fig:add_noise}
\end{figure}

\subsubsection{Multi-Ciphertext Watermarking}

Since the noise in \pkscheme ciphertexts is relatively larger, we focus our discussion on \pkscheme. In \cref{fig:mw}, we show the standard deviation of $\rho$ for $I_w = 10$ and $I_w = 100$. From \cref{fig:intensity_pk}, it can be seen that when $I_w = 10$, the standard deviation of $\rho$ is similar across different embedding cases. This indicates that the noise-reduction effect of multi-ciphertext watermarking is consistent across different embedding conditions. 
As the number of ciphertexts increases, the standard deviation of $\rho$ decreases, reducing the interference of extraction noise on the watermark. 
For a larger $I_w$, the number of ciphertexts has a greater impact on the extracted $\rho$ when embedding information, as shown in \cref{fig:mw_100}.



\begin{figure}[!t]
	\centering  
	\subfigbottomskip=1pt 
	\subfigcapskip=1pt 
    \subfigure[$I_w=10$]{
    \label{fig:mw_10}
    \includegraphics[width=0.45\linewidth]{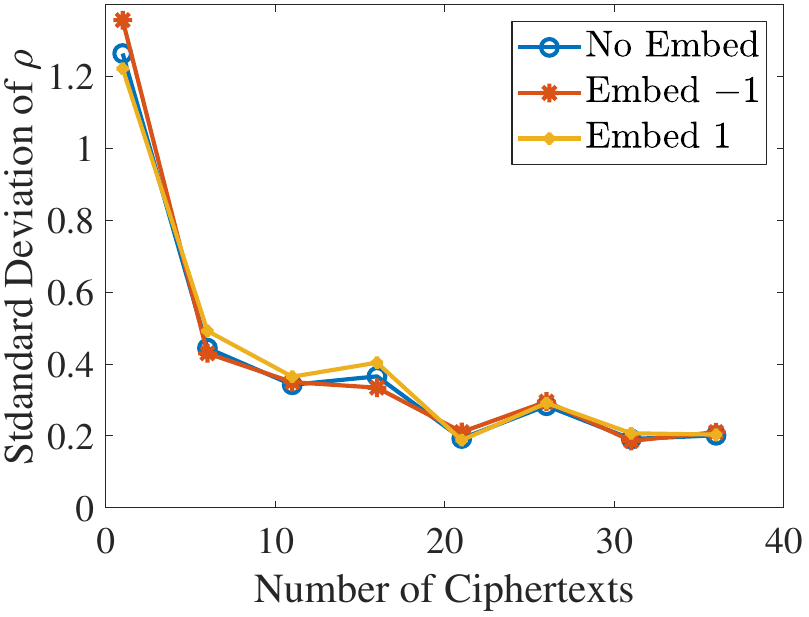}}
	\subfigure[$I_w=100$]{
    \label{fig:mw_100}
    \includegraphics[width=0.45\linewidth]{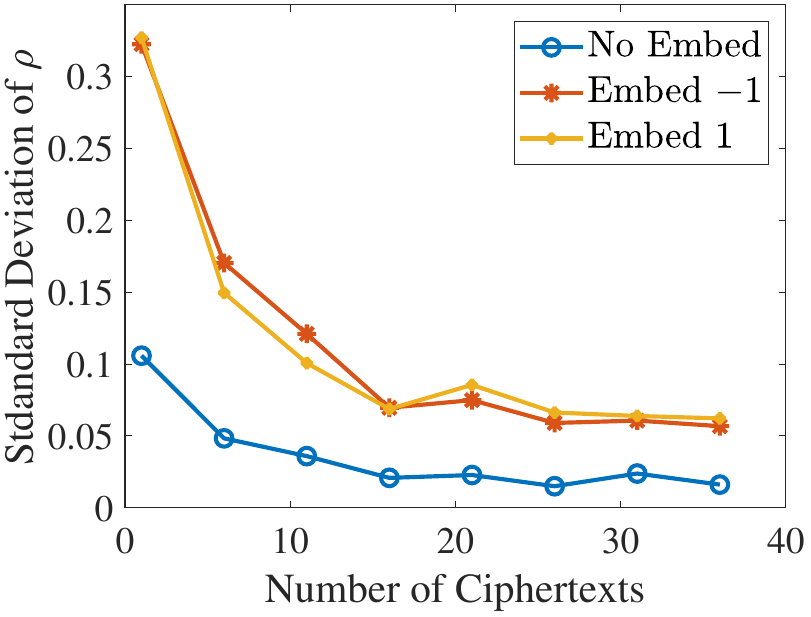}}
	\caption{
 Standard deviation of $\rho$ with different number of ciphertexts.
}
\label{fig:mw}
\end{figure}

\subsection{Evaluation for \MultRWatermark}
Due to the larger noise in \pkscheme ciphertexts, we only consider the case where the embedding carrier is a \skscheme ciphertext. We set $N = 2048$ and $p = 131$. During the experiments, plaintexts are sampled uniformly from the interval $[0, p)$. The number of ciphertexts is set to $m = 4$. The TPR of watermark extraction is estimated through $1000$ independent trials.

\subsubsection{Correctness Verification}
    

\begin{table}[!t]
    \caption{TPR of \MultRWatermark with different intensities.}
    \centering
    \begin{tabular}{ccccccc}
    \toprule
        $I_w$ & 24& 26& 28& 30& 32& 34  \\ \midrule
         TPR& 0.172& 0.652& 0.885& 0.975& 0.994& 1.000 \\
         \bottomrule
    \end{tabular}
    
    \label{tab:le_correct}
\end{table}

In \cref{tab:le_correct}, we present the TPR under different values of $I_w$. When $I_w$ is small, the encryption noise may overlap with the watermark noise, thereby corrupting the watermark. As $I_w$ increases, the TPR gradually improves. When $I_w = 34$, the embedded watermark can be extracted perfectly.

\subsubsection{Robustness to Addition Operations}

We evaluate the robustness of \MultRWatermark against homomorphic addition by adding ciphertexts without embedded watermarks to watermarked ciphertexts. In \cref{fig:add_le_test}, we present the experimental results. Adding additional ciphertexts increases the noise level; consequently, when $I_w$ is small, the extracted TPR is relatively low. As the number of addition operations increases, $I_w$ should be increased accordingly to maintain an acceptable TPR.

\subsubsection{Resistance to Random Noise Attack}

In \cref{fig:add_le_attack}, we present the impact of adding random noise. Similar to the attacks on \AddRWatermark, the noise added here is random but does not affect decryption. The noise coefficients are uniformly sampled from $[-B_a, B_a]$. When $B_e$ is low, watermark extraction is barely affected; however, as the noise increases, the watermark is increasingly impacted. 
As shown in \cref{fig:add_le_attack}, increasing $I_w$ can resist stronger noise. Compared with \AddRWatermark, the \MultRWatermark watermark is more sensitive to noise. As the noise increases, the attack success rate against \MultRWatermark rises rapidly. Therefore, \MultRWatermark generally requires a higher $I_w$.


\begin{figure}[!t]
	\centering  
	\subfigbottomskip=1pt 
	\subfigcapskip=1pt 
    \subfigure[HE additions]{
    \label{fig:add_le_test}
    \includegraphics[width=0.45\linewidth]{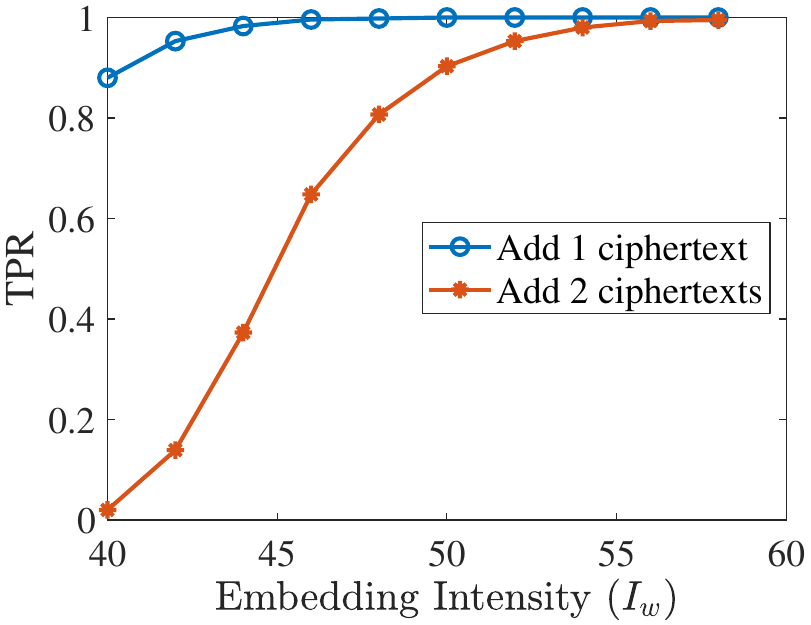}}
	\subfigure[Adding Noise]{
    \label{fig:add_le_attack}
    \includegraphics[width=0.45\linewidth]{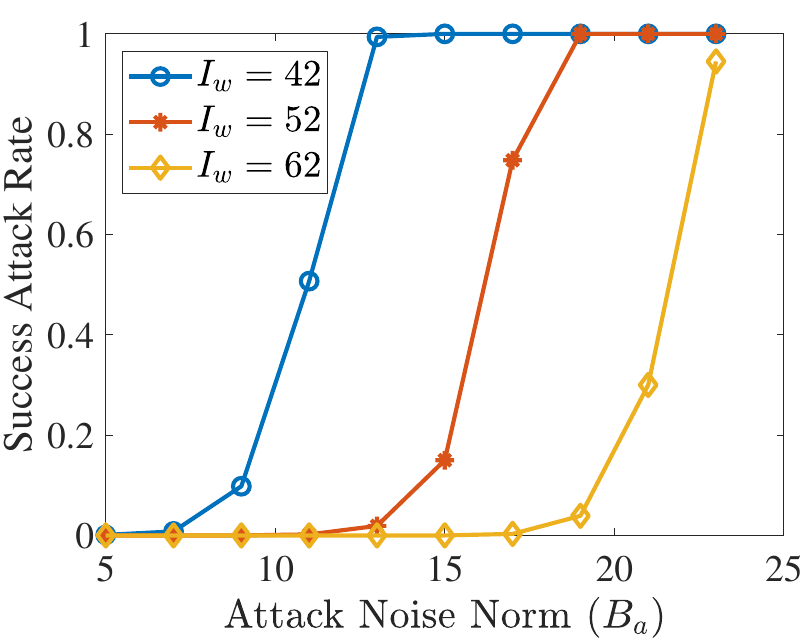}}
	\caption{
 Effect of addition operations on \MultRWatermark.
}
\label{fig:add_le}
\end{figure}

\subsubsection{Robustness to Multiplication Operations}

For homomorphic multiplication, we multiply a watermarked ciphertext by a ciphertext without a watermark. In addition, we consider the effect of the $\Resize$ operation. As shown in \cref{tab:le_mult_test}, maintaining watermark robustness under multiplication operations requires a relatively large $I_w$, since multiplication introduces additional noise. Even without applying $\Resize$, the watermark can be reliably extracted only when $I_w \geq 782327$. 

Since the $\Resize$ operation introduces further noise, a larger $I_w$ is required to ensure correct watermark extraction. In the experiments, we set the parameter of $\Resize$ to $T = 65537$. The results show that the TPR continuously improves as $I_w$ increases. When $I_w = 782327$, perfect extraction can be achieved without $\Resize$, while for ciphertexts processed with $\Resize$, perfect extraction is achieved when $I_w = 95617843$. Therefore, the use of $\Resize$ requires $I_w$ to be approximately two orders of magnitude larger to maintain reliable extraction performance.


\begin{table}[!t]
    \caption{TPR of \MultRWatermark with different intensities when multiplying by one ciphertext.}
    \centering
    \begin{tabular}{@{}ccccc@{}}
    \toprule
    \multirow{2}{*}{No $\Resize$}&
        $I_w$ &  608477& 695402& 782327\\ \cmidrule{2-5}
         &TPR & 0.867& 0.992& 1.000\\ \midrule
    \multirow{2}{*}{$\Resize$}&
        $I_w$ & 78232780& 86925312& 95617843\\   \cmidrule{2-5}
         &TPR&  0.966& 0.992& 1.000\\
         \bottomrule
    \end{tabular}
    
    \label{tab:le_mult_test}
\end{table}

\subsection{Overhead Comparison}
\begin{table}[!t]
\centering
\caption{Overhead (ms) of \skscheme ciphertext watermarking.}
\label{tab:time_overhead}
\begin{tabular}{ccccc}
\toprule
 & Encrypt & Embed & Decrypt & Extract \\ \midrule
\AddRWatermark & 282.77 & 0.83 & 254.74 & 255.71 \\
\MultRWatermark & 271.83 & 0.77 & 267.00 & 267.75 \\ \bottomrule
\end{tabular}
\end{table}

We show the average overhead of watermarking a single \skscheme ciphertext in \cref{tab:time_overhead}. 
The embedding and extraction overheads of \AddRWatermark and \MultRWatermark are comparable. 
The embedding operation takes less than 1 ms, indicating that the additional cost introduced by watermark embedding is negligible. 
The extraction operation incurs a latency close to that of decryption, as it requires decrypting the ciphertext before verifying the embedded watermark information, with only a few additional milliseconds for watermark confirmation. 
Furthermore, our watermarking schemes introduce no additional computational overhead for homomorphic evaluations.

\section{Conclusion}

In this work, we introduce the watermarking technique into the domain of HE to enable ciphertext provenance, ownership protection, and tracing of ciphertexts throughout homomorphic computations. Homomorphic ciphertext watermarking provides protection for HE ciphertexts from both the data provider's and the evaluator's perspectives, thereby extending the application scenarios of HE and facilitating its practical deployment.
We provide a formal definition of homomorphic ciphertext watermarking and propose two concrete watermarking schemes for RLWE ciphertexts, demonstrating the feasibility of homomorphic ciphertext watermarking. 
However, the two proposed schemes have their limitations. \AddRWatermark is not robust under homomorphic multiplication, whereas \MultRWatermark has a relatively small embedding capacity. 
In future work, we aim to address these limitations and design more practical homomorphic ciphertext watermarking schemes. Additionally, we plan to explore broader application scenarios for homomorphic ciphertext watermarking to provide stronger security guarantees for various privacy-preserving computation tasks.

\bibliographystyle{plainurl}
\bibliography{ref.bib}
\appendix
\section{Ethical Considerations}

This work proposes a watermarking technique for RLWE-based HE ciphertexts, aiming to provide copyright protection, provenance tracking, and computation supervision for privacy-preserving computation. The proposed approach does not access or reveal the underlying plaintexts, and the watermark embedding process is designed to preserve the security and correctness of the underlying HE scheme.

As with other watermarking and attribution technologies, our approach may have dual-use implications. For example, while ciphertext watermarking can help prevent unauthorized use and support accountability in outsourced computation, it could potentially be misused for unauthorized tracking or attribution. We believe that the deployment of such techniques should follow appropriate transparency, consent, and governance practices. This work focuses on the technical feasibility and security properties of ciphertext watermarking and does not involve collecting user data, interacting with human subjects, or evaluating sensitive applications.





    
    
    



\section{Additional Description of \AHE}
\label{appendix:ahe}
\subsection{\pkscheme of \AHE}
The main difference between \pkscheme and \skscheme is that \pkscheme requires generating a public key from the secret key. Encryption is then performed using the public key instead of the secret key. The remaining procedures of \pkscheme are identical to those of \skscheme.
\begin{itemize}
    \item $\pkGen(\sk)$: Sample $\mathbf{a}$ uniformly from $\mathbf{R}_q$ and $\mathbf{e}$ from error distribution $\chi_e$. Compute $\mathbf{b}=\mathbf{a}\cdot \mathbf{s}+p\mathbf{e}$. Output the public key $\pk=(\mathbf{b},-\mathbf{a})$. The public key is actually a \skscheme ciphertext of zero.
    
    \item $\pkEn(\pk,\mathbf{m})$: Let $\pk=(\mathbf{k}_1,\mathbf{k}_2)$. Sample $\mathbf{u}$ randomly from secret key distribution $\chi_s$. Sample $\mathbf{e}_1$ and $\mathbf{e}_2$ from error distribution $\chi_e$. Output the ciphertext
    $(\mathbf{k}_1\cdot \mathbf{u}+p\mathbf{e}_1+\mathbf{m},\mathbf{k}_2\cdot \mathbf{u}+p\mathbf{e}_2)$.
\end{itemize}

Next, we extend \AHE to support a bounded number of homomorphic multiplications.

\begin{itemize}
    \item $\EvalMult(c_t,c_t^\prime)$: Let $c_t=(\mathbf{c}_1,\mathbf{c}_2)$ be the ciphertext of $\mathbf{m}_1$ and $c_t^\prime=(\mathbf{c}^\prime_1,\mathbf{c}^\prime_2)$ be the ciphertext of $\mathbf{m}_2$. Output the ciphertext $(\mathbf{c}_1 \cdot \mathbf{c}^\prime_1,\mathbf{c}_1 \cdot \mathbf{c}^\prime_2+\mathbf{c}^\prime_1 \cdot \mathbf{c}_2, \mathbf{c}_2 \cdot \mathbf{c}^\prime_2)$, which is a ciphertext of the product of $\mathbf{m}_1$ and $\mathbf{m}_2$.

    \item $\DeT(\sk,c_t)$: Let $c_t=(\mathbf{d}_1,\mathbf{d}_2,\mathbf{d}_3)$. Output the plaintext $\mathbf{m}=(\mathbf{d}_1+\mathbf{s}\cdot \mathbf{d}_2+ \mathbf{s}^2 \cdot \mathbf{d}_3) \mod p$.
\end{itemize}

The relinearization technique proposed in~\cite{brakerski2014leveled} can be applied to reduce the size of the ciphertext after multiplication to the same size as the original ciphertext, allowing decryption using $\De$. As long as the noise remains sufficiently small, multiple homomorphic multiplications can be supported.

\begin{itemize}
    \item $\EvkGen(\sk)$: Let $T$ be a positive integer, generate $\lfloor \log_T q \rfloor +1$ ciphertexts. Sample $\mathbf{a}_i$ uniformly from $\mathbf{R}_q$, sample $\mathbf{e}_i$ from error distribution $\chi_e$, 
    and compute $c_{t,i} = (T^i\mathbf{a}_i \cdot \mathbf{s}+p\mathbf{e}_i+T^i\mathbf{s}^2,-T^i\mathbf{a}_i)$, where $i=1,2,\cdots, \lfloor \log_T q \rfloor +1$. Output the evaluation key $\evk=(c_{t,0},c_{t,1},\cdots, c_{t,{\lfloor \log_T q \rfloor +1}})$.

    \item $\Resize(\evk, (\mathbf{d}_1,\mathbf{d}_2,\mathbf{d}_3))$: Decompose $\mathbf{d}_3$ into its base-$T$ representation:
    $\mathbf{d}_3=\sum_{i=0}^{\lfloor \log_T q \rfloor +1}T^i\mathbf{d}_{3,i}$,
    where the coefficients of $\mathbf{d}_{3,i}$ have absolute values less than $T$. 
    Compute $\mathbf{c}^\prime_1=\mathbf{d}_1+\sum_{i=0}^{\lfloor \log_T q \rfloor +1}\mathbf{d}_{3,i}\cdot c_{t,i}[1], \mathbf{c}^\prime_2=\mathbf{d}_2+\sum_{i=0}^{\lfloor \log_T q \rfloor +1}\mathbf{d}_{3,i}\cdot c_{t,i}[2]$. Output the new ciphertext $(\mathbf{c}^\prime_1,\mathbf{c}^\prime_2)$.
\end{itemize}

\subsection{Security of \AHE}
\label{sec:HE_secure_proof}

We focus on security against chosen-plaintext attacks (CPA security). We therefore first review the corresponding definition below~\cite{katz2015introduction}.

\begin{definition}
\label{def:CPA}
Assume that $b$ is sampled uniformly at random from $\{0,1\}$.
Let $c_t$ denote a ciphertext of $\mathbf{m}_b$ generated by an encryption scheme. The challenger gives $c_t$ to an adversary. The adversary is allowed to query the encryption oracle on any plaintext. 
If every probabilistic polynomial-time adversary can correctly identify b with probability at most \(1/2 + \epsilon(\lambda)\), where \(\epsilon(\lambda)\) is negligible, then the encryption scheme is CPA-secure.
\end{definition}

The security of \AHE is based on the hardness of the RLWE problem. We state the following theorem.

\begin{theorem}
\label{th:CPA-secure}
    The encryption scheme \AHE defined in \cref{sec:homomorphic_encryption} is CPA-secure under the hardness assumption of the RLWE problem.
\end{theorem}

Our proof consists of two steps. We first show that \skscheme is CPA-secure. Based on this result, we then prove that \pkscheme is CPA-secure. 




The proof of the security of \skscheme is given as follows:
\begin{proof}
    Let $c_{t,0}=\skEn(\sk,\mathbf{m}_0)=(\mathbf{a}_0 \cdot \mathbf{s}+p\mathbf{e}_0+\mathbf{m}_0,-\mathbf{a}_0)$ and  $c_{t,1}=\skEn(\sk,\mathbf{m}_1)=(\mathbf{a}_1 \cdot \mathbf{s}+p\mathbf{e}_1+\mathbf{m}_1,-\mathbf{a}_1)$.
    Let $(\mathbf{b}_r,\mathbf{a}_r)$ be uniformly random.

    Since $\gcd(p,q)=1$, the inverse $p^{-1}$ exists modulo $q$ and is a constant. Therefore, $\mathbf{a}_0$ and $p^{-1}\mathbf{a}_0$ follow the same distribution $\chi_u$. Because of the hardness of RLWE, the two pairs $(p^{-1}\mathbf{a}_0 \cdot \mathbf{s}+\mathbf{e},-p^{-1}\mathbf{a})$ and $(\mathbf{b}_r,\mathbf{a}_r)$ are indistinguishable. Because $p$ and $\mathbf{m}_0$ are constants, $(\mathbf{a}_0 \cdot \mathbf{s}+p\mathbf{e}_0+\mathbf{m}_0,-\mathbf{a}_0)$ and $(\mathbf{b}_r,\mathbf{a}_r)$ are indistinguishable. That is to say $c_{t,0}$ is indistinguishable from $(\mathbf{b}_r,\mathbf{a}_r)$. The same argument shows that $c_{t,1}$ is indistinguishable from $(\mathbf{b}_r,\mathbf{a}_r)$. Thus, $c_{t,0}$ and $c_{t,1}$ are indistinguishable, which means the encryption scheme \skscheme defined in \cref{sec:homomorphic_encryption} is CPA-secure.
\end{proof}

We next prove that \pkscheme is CPA-secure:
\begin{proof}
    The public key $\pk$ is a ciphertext of \skscheme. Because \skscheme is CPA-secure, the public key $\pk$ is indistinguishable from uniformly randomly sampled $(\mathbf{b}_r,\mathbf{a}_r)$.  
    Suppose that $c_{t,0}=\pkEn(\pk,\mathbf{m}_0)=(\mathbf{k}_0\cdot \mathbf{u}+p\mathbf{e}_1+\mathbf{m}_0,\mathbf{k}_1\cdot \mathbf{u}+p\mathbf{e}_2)$, where $\mathbf{u} \sim \chi_s$ and  $\mathbf{e}_1,\mathbf{e}_2 \sim \chi_e$. Actually, the public key $\pk$ and $c_{t,0}$ are two ciphertexts of \skscheme.
    $(\mathbf{u}+p\mathbf{e}_1+\mathbf{m},-\mathbf{k}_0)$,$(\mathbf{k}_1\cdot \mathbf{u}+p\mathbf{e}_2,-\mathbf{k}_1)$. 
    Hence, we have $c_{t,0}$ is indistinguishable from $(\mathbf{b}_r,\mathbf{a}_r)$.
    Let $c_{t,1}=\pkEn(\pk,\mathbf{m}_1)$. We can prove that $c_{t,1}$ is indistinguishable from $(\mathbf{b}_r,\mathbf{a}_r)$ following the same steps. Thus, the two ciphertexts $c_{t,0}$ and $c_{t,1}$ are indistinguishable. Hence, the encryption scheme \pkscheme defined in \cref{sec:homomorphic_encryption} is CPA-secure.
\end{proof}

\section{Lemmas}
\label{appendix:lemmas}
\begin{lemma}
\label{lm:norm_mul}
Let $\mathbf{x} \in \mathbf{R}_q$, $\mathbf{y} \in \mathbf{R}_q$. Then
$\Vert \mathbf{x}\cdot \mathbf{y} \Vert \leq N\Vert \mathbf{x} \Vert \cdot \Vert \mathbf{y} \Vert$.
\end{lemma}

\begin{proof}
Multiplication in $\mathbf{R}_q$ is performed modulo the polynomial $x^N + 1$, hence $x^N = -1$ in $\mathbf{R}_q$.

Let $\mathbf{z} = \mathbf{x} \cdot \mathbf{y}$. According to the arithmetic in $\mathbf{R}_q$, the $i$-th coefficient of $\mathbf{z}$ can be written as
$\mathbf{z}[i] = \sum_{j+k=i} \mathbf{x}[j]\mathbf{y}[k] - \sum_{j+k=N+i} \mathbf{x}[j]\mathbf{y}[k]$,
where $i,j,k \in \{0,1,\dots,N-1\}$.

For a fixed $i$, the equation $j+k=i$ admits $i+1$ valid pairs $(j,k)$. Since $j,k < N$, the equation $j+k=N+i$ admits $N-i-1$ valid pairs, namely $j \in \{i+1,\dots,N-1\}$. Therefore, each coefficient $\mathbf{z}[i]$ is a sum of exactly $N$ products of the form $\mathbf{x}[j]\mathbf{y}[k]$.

It follows that
$|\mathbf{z}[i]| \leq N \cdot \max_j |\mathbf{x}[j]| \cdot \max_k |\mathbf{y}[k]|$.
Consequently,
$\Vert \mathbf{z} \Vert \leq N\Vert \mathbf{x} \Vert \cdot \Vert \mathbf{y} \Vert$.
\end{proof}

\begin{lemma}
\label{lm:xy_e_v}
    Let $x$ and $y$ be two independent random variables. If $\E(x)=\E(y)=0$, then $\E(xy)=0,\var(xy)=\var(x)\cdot \var(y)$.
\end{lemma}

\begin{proof}
    Since $x$ and $y$ are independent, we have $\E(xy)=\E(x)\cdot \E(y)=0$.
    Moreover, $x^2$ and $y^2$ are also independent, so $\E(x^2y^2)=\E(x^2)\cdot \E(y^2)$.
    The variance of $xy$ is:
    $\var(xy)=\E(x^2y^2)-\E(xy)^2
        =\E(x^2)\cdot \E(y^2)
        =(\E(x^2)-0)\cdot (\E(y^2)-0)
        =\var(x)\cdot \var(y).
    $
\end{proof}

\begin{lemma}
\label{lm:xy_e}
Let $\mathbf{x},\mathbf{y} \sim \chi_e$, and assume that $\mathbf{x}$ and $\mathbf{y}$ are independent. Then
$\E(\langle \mathbf{x},\mathbf{y} \rangle)=0$ and $\var(\langle \mathbf{x},\mathbf{y} \rangle)=N\sigma^4$.
\end{lemma}

\begin{proof}
The components $\mathbf{x}[i]$ and $\mathbf{y}[i]$ are mutually independent, with zero mean and variance $\sigma^2$. By \cref{lm:xy_e_v}, we have $\E(\mathbf{x}[i]\cdot \mathbf{y}[i])=0$ and $\var(\mathbf{x}[i]\cdot \mathbf{y}[i])=\sigma^4$.
Hence, we have
$\E(\langle \mathbf{x},\mathbf{y} \rangle)=\E\left(\sum_i \mathbf{x}[i]\cdot \mathbf{y}[i]\right)
=\sum_i \E(\mathbf{x}[i]\cdot \mathbf{y}[i])
=0$
and
$\var(\langle \mathbf{x},\mathbf{y} \rangle)=\var\left(\sum_i \mathbf{x}[i]\cdot \mathbf{y}[i]\right)
=\sum_i \var(\mathbf{x}[i]\cdot \mathbf{y}[i])
= N\sigma^4$.
\end{proof}

\begin{lemma}
\label{lm:xx_e}
Let $\mathbf{x}\sim \chi_e$. Then $\E(\langle \mathbf{x},\mathbf{x}\rangle)=N\sigma^2$ and $\var(\langle \mathbf{x},\mathbf{x}\rangle)=2N\sigma^4$.
\end{lemma}
\begin{proof}
Each component $\mathbf{x}[i]\sim \mathcal{N}(0,\sigma^2)$, hence $\frac{\mathbf{x}[i]}{\sigma}$ follows the standard normal distribution. Therefore, $\langle \frac{\mathbf{x}}{\sigma},\frac{\mathbf{x}}{\sigma}\rangle$ follows a chi-square distribution with $N$ degrees of freedom.
We have
$\E(\langle\frac{\mathbf{x}}{\sigma},\frac{\mathbf{x}}{\sigma}\rangle)=N$ and $\var(\langle \frac{\mathbf{x}}{\sigma},\frac{\mathbf{x}}{\sigma}\rangle)=2N$.
Thus,
$\E(\langle \mathbf{x},\mathbf{x}\rangle)=\E(\sigma^2 \cdot \langle \frac{\mathbf{x}}{\sigma},\frac{\mathbf{x}}{\sigma}\rangle )=\sigma^2 \cdot \E(\langle \frac{\mathbf{x}}{\sigma},\frac{\mathbf{x}}{\sigma}\rangle)=N\sigma^2.$
Similarly,
$\var(\langle\mathbf{x},\mathbf{x}\rangle)=\var(\sigma^2 \cdot \langle\frac{\mathbf{x}}{\sigma},\frac{\mathbf{x}}{\sigma}\rangle )=\sigma^4 \cdot \var(\langle\frac{\mathbf{x}}{\sigma},\frac{\mathbf{x}}{\sigma}\rangle)=2N\sigma^4$.
\end{proof}

\begin{lemma}
\label{lm:cov}
Let $x,y$ be independent, and $x,z$ be independent. If $\E(x)=\E(y)=0$, then $\cov(xy,yz)=0$. 
\end{lemma}

\begin{proof}
Since $x$ and $y$ are independent, we have $\E(xy)=\E(x)\cdot \E(y)=0$.
Moreover, since $x$ and $z$ are independent, $\E(xy[yz-\E(yz)])=\E(x)\cdot \E(y[yz-\E(yz)])=0$. Hence,
$\cov(xy,yz)=\E([xy-\E(xy)]\cdot[yz-\E(yz)])
=\E(xy\cdot[yz-\E(yz)])
=\E(x)\cdot \E(y\cdot[yz-\E(yz)])=0$.
\end{proof}

\begin{lemma}
\label{lm:r_xy}
Let $\mathbf{x} \sim \chi_s$ and $\mathbf{y} \sim \chi_e$. Let $\mathbf{z}=\mathbf{x}\cdot \mathbf{y}$. Then $\E(\mathbf{z}[i])=0$ and $\var(\mathbf{z}[i])=\frac{2N\sigma^2}{3}$.
\end{lemma}

\begin{proof}
    The polynomial modulus of $\mathbf{R}_q$ is $x^N+1$, which implies that $x^N=-1$ in $\mathbf{R}_q$. Hence,
    $\mathbf{z}[i]=\sum_{j+k=i}\mathbf{x}[j]\cdot \mathbf{y}[k]-\sum_{j+k=i+N}\mathbf{x}[j]\cdot \mathbf{y}[k]$.
    Since the ciphertext modulus $q$ is much larger than $B_e$, modular reduction modulo $q$ can be ignored in the analysis.
    Because $\mathbf{x}[j]$ and $\mathbf{y}[k]$ are independent, by \cref{lm:xy_e_v}, we have
    $\E(\mathbf{x}[j]\cdot \mathbf{y}[k])=0$
    and
    $\var(\mathbf{x}[j]\cdot \mathbf{y}[k])=\frac{2\sigma^2}{3}$.
    Hence, $\E(\mathbf{z}[i])=\E(\sum_{j+k=i}\mathbf{x}[j]\cdot \mathbf{y}[k]-\sum_{j+k=i+N}\mathbf{x}[j]\cdot \mathbf{y}[k])
    =\sum_{j+k=i}\E(\mathbf{x}[j]\cdot \mathbf{y}[k])-\sum_{j+k=i+N}\E(\mathbf{x}[j]\cdot \mathbf{y}[k])
    =0$.
    By \cref{lm:cov}, when $j+k \neq j^\prime+k^\prime$,
    $\cov(\mathbf{x}[j]\cdot \mathbf{y}[k],\mathbf{x}[j^\prime]\cdot \mathbf{y}[k^\prime])=0$.
    Hence,
    $\var(\mathbf{z}[i])=\var(\sum_{j+k=i}\mathbf{x}[j]\cdot \mathbf{y}[k]-\sum_{j+k=i+N}\mathbf{x}[j]\cdot \mathbf{y}[k])
    =\sum_{j+k=i}\var(\mathbf{x}[j]\cdot \mathbf{y}[k])+\sum_{j+k=i+N}\var(\mathbf{x}[j]\cdot \mathbf{y}[k])
    =\frac{2N\sigma^2}{3}$.
\end{proof}

\section{Other Analysis of \AddRWatermark for \skscheme}
\label{appendix:other_analysis_add}
\subsection{Security-Preserving}
Our embedding scheme does not affect the security of the underlying HE scheme. Formally, we have the following theorem:
\begin{theorem}
\label{th:add_passive_attack}
    The watermarking scheme \AddRWatermark defined in \cref{sec:add_watermark}  is security-preserving.
\end{theorem}
\begin{proof}
    Let $c_{t,0}=(\mathbf{c}_{0,1},\mathbf{c}_{0,2})$ and $c_{t,1}=(\mathbf{c}_{1,1},\mathbf{c}_{1,2})$ be two \AHE ciphertexts. Let $(\mathbf{b}_r,\mathbf{a}_r)$ be uniformly random. Let $\mathbf{k}_w$ be the $\emk$, $w_e$ be the encoded watermark and $I_w$ be the embedding intensity.
    
    According to \cref{th:CPA-secure}, we know that $c_{t,0}$ and $c_{t,1}$ are indistinguishable from $(\mathbf{b}_r,\mathbf{a}_r)$. So, it is enough to prove that $(\mathbf{c}_{b,1}+pI_w\mathbf{k}_w, \mathbf{c}_{b,2})$ is indistinguishable from $(\mathbf{b}_r,\mathbf{a}_r)$, where $b$ is $0$ or $1$. Since $\mathbf{c}_{b,1}$ is  indistinguishable from $\mathbf{b}_r$ and $pI_w\mathbf{k}_w$ is a constant element of $\mathbf{R}_q$, we have that $\mathbf{c}_{b,1}+pI_w\mathbf{k}_w$ is indistinguishable from $\mathbf{c}_{b,1}$. Hence, we have that $(\mathbf{c}_{b,1}+pI_w\mathbf{k}_w, \mathbf{c}_{b,2})$ is indistinguishable from $(\mathbf{b}_r,\mathbf{a}_r)$.
\end{proof}

In the following, we focus on the watermark embedded in \skscheme ciphertexts and defer the analysis of watermark embedded in \pkscheme ciphertexts to Appendix~\cref{appendix:analysis_of_pk_watermark}.

\subsection{Fidelity}



For ciphertext watermark embedding in \skscheme, the following theorem holds:
\begin{theorem}
\label{th:sk_noise_bound}
When the watermark carrier is a ciphertext of \skscheme, and the embedding intensity satisfies $I_w<\frac{q}{2pB_e}-2$, the homomorphic ciphertext watermarking scheme \AddRWatermark described in \cref{sec:add_watermark} achieves fidelity.
\end{theorem}

\begin{proof}
Let the ciphertext encrypted under \skscheme be
$c_t=(\mathbf{a}\cdot \mathbf{s}+p\mathbf{e}_1+\mathbf{m},-\mathbf{a})$,
where $\mathbf{a} \sim \chi_u$ and $\mathbf{e}_1 \sim \chi_e$.
According to the embedding rule in \cref{eq:s_insert_watermark}, the ciphertext after embedding $w_e$ is given by
$c_t^\prime=(\mathbf{a}\cdot \mathbf{s}+p\mathbf{e}_1+\mathbf{m}+pw_eI_w\mathbf{k}_w,-\mathbf{a})$.

The inner product of $c_t^\prime$ with the secret key $\sk$ is:
\begin{equation}
\label{eq:noise_in_sk}
    \begin{aligned}
        \langle \sk,c_t^\prime\rangle
        &=\mathbf{a}\cdot \mathbf{s}+p\mathbf{e}_1+\mathbf{m}+pw_eI_w\mathbf{k}_w-\mathbf{a}\cdot \mathbf{s}\\
        &=p\mathbf{e}_1+\mathbf{m}+pw_eI_w\mathbf{k}_w.
    \end{aligned}
\end{equation}

By applying the norm inequality, we have
$\left\lVert \langle \sk,c_t^\prime\rangle \right\rVert
= \left\lVert p\mathbf{e}_1 + pw_eI_w\mathbf{k}_w + \mathbf{m} \right\rVert
\leq pB_e + pI_wB_e + p$.
Since $p \ll q$, and $B_e$ is a small integer determined by the security parameter, when $I_w<\frac{q}{2pB_e}-2$, it holds that
$\left\lVert \langle \sk,c_t^\prime\rangle \right\rVert < \frac{q}{2}$.
Therefore, $c_t^\prime$ can be correctly decrypted, and we obtain
$\mathbf{m}=\langle \sk,c_t^\prime\rangle \mod p$.
\end{proof}






\section{Other Analysis of \MultRWatermark}
\label{appendix:other_analysis_mult}
According to the construction method of the linear equation system described in \cref{sec:A_le_construct}, we can control the norm of the initially embedded solution. Therefore, we assume that the solution norm satisfies $\Vert X_{s,i} \Vert < B_e$, i.e., its norm satisfies the same requirement as $\mathbf{k}_w$ in \AddRWatermark.
Since \MultRWatermark is a special case of \AddRWatermark, it provides the same security and fidelity guarantees. Note that we only consider \skscheme ciphertexts in this setting.
The security analysis of \AddRWatermark does not depend on $\mathbf{k}_w$, while the fidelity analysis only relies on the norm bound of $\mathbf{k}_w$.
Therefore, \MultRWatermark also satisfies the conclusions of \cref{th:add_passive_attack} and \cref{th:sk_noise_bound}.
Accordingly, in the following, we focus only on the robustness of \MultRWatermark.

\subsection{Robustness}

When $I_w$ is larger than the number of homomorphic additions, we have $\lfloor \frac{\mathbf{m}}{I_wp}\rceil=0$. Therefore, the plaintext hidden in the ciphertext does not affect the extracted correlation value. In the following, we consider adding a ciphertext with a watermark embedded using a different key.

\begin{theorem}
\label{th:sk_r_2}
Let \skscheme ciphertext $c_{t,0}=(\mathbf{a}_0\cdot \mathbf{s}+p\mathbf{e}_0+\mathbf{m}_0+pw_eI_w\mathbf{k}_w,-\mathbf{a}_0)$ be the ciphertext after embedding $w_e$ with embedding key $\mathbf{k}_w$, and let $c_{t,1}=(\mathbf{a}_1\cdot \mathbf{s}+p\mathbf{e}_1+\mathbf{m}_1+pw_e^\prime I_w^\prime \mathbf{k}_w^\prime,-\mathbf{a}_1)$ be the ciphertext after embedding $w_e^\prime$ with embedding key $\mathbf{k}_w^\prime$.
The embedding key $\mathbf{k}_w$ is independent of $\mathbf{k}_w^\prime$.
Let $\rho_{sk}$ denote the correlation value of $c_{t,0}$, and let $\rho^\prime_{sk}$ denote the correlation value of $c_{t,0}+c_{t,1}$ when taking $\mathbf{k}_w$ as extracting key.
Then, $\E(\rho_{sk})=\E(\rho^\prime_{sk})$, and
$
\var(\rho^\prime_{sk})=\frac{\sigma^4}{NI_w^2}+\frac{w_e^{\prime 2} I_w^{\prime 2}\sigma^4}{NI_w^2}+\var(\rho_{sk}).
$
\end{theorem}

\begin{proof}
    The ciphertext sum is given by$c_{t,0}+c_{t,1}=((\mathbf{a}_0+\mathbf{a}_1)\cdot \mathbf{s}+p(\mathbf{e}_0+\mathbf{e}_1)+\mathbf{m}_0+\mathbf{m}_1+pw_eI_w\mathbf{k}_w+pw_e^\prime I_w^\prime \mathbf{k}_w^\prime,-(\mathbf{a}_0+\mathbf{a}_1)).$ 
    Using \cref{eq:extract_noise}, the extracted noise leads to the following correlation value:
    \begin{equation*}
         \rho=\frac{1}{NI_w}\cdot \langle \mathbf{e}_0+\mathbf{e}_1+w_eI_w\mathbf{k}_w+w_e^\prime I_w^\prime \mathbf{k}_w^\prime,\mathbf{k}_w\rangle.
    \end{equation*}

    Note that $\mathbf{e}_1$ is independent of $\mathbf{k}_w$.
    Based on this observation, the expectation of $\rho^\prime_{sk}$ can be computed as
    \begin{equation*}
        \begin{aligned}
        \E(\rho^\prime_{sk})&=\frac{1}{NI_w}\E(\langle\mathbf{e}_0+\mathbf{e}_1+w_e^\prime I_w^\prime \mathbf{k}_w^\prime+w_eI_w\mathbf{k}_w,\mathbf{k}_w\rangle)\\
        &=\frac{\E(\langle\mathbf{e}_1,\mathbf{k}_w\rangle)}{NI_w}+\frac{w_e^\prime I_w^\prime }{NI_w}\E(\langle\mathbf{k}_w^\prime,\mathbf{k}_w\rangle)+\E(\rho_{sk})\\
        &=\E(\rho_{sk}).    
        \end{aligned}
    \end{equation*}
    
Similarly, the variance of $\rho^\prime_{sk}$ is given by
\begin{equation*}
    \begin{aligned}
    \var(\rho^\prime_{sk})&=\frac{1}{N^2I_w^2}\var(\langle\mathbf{e}_0+\mathbf{e}_1+w_e^\prime I_w^\prime \mathbf{k}_w^\prime+w_eI_w\mathbf{k}_w,\mathbf{k}_w\rangle)\\
    &=\frac{\var(\langle\mathbf{e}_1,\mathbf{k}_w\rangle)}{N^2I_w^2}+\frac{w_e^{\prime 2} I_w^{\prime 2} }{N^2I_w^2}\var(\langle\mathbf{k}_w^\prime,\mathbf{k}_w\rangle)+\var(\rho_{sk})\\
    &=\frac{\sigma^4}{NI_w^2}+\frac{w_e^{\prime 2} I_w^{\prime 2}\sigma^4}{NI_w^2}+\var(\rho_{sk}).   
    \end{aligned}
\end{equation*}

\end{proof}




Therefore, for \AddRWatermark watermarked \skscheme ciphertexts, adding a ciphertext that also contains a watermark $w_e^\prime$ increases the variance of $\rho$ by $\frac{\sigma^4}{NI_w^2}+\frac{w_e^{\prime 2} I_w^{\prime 2}\sigma^4}{NI_w^2}$.
If the added ciphertext does not contain a watermark, i.e., $w_e^\prime=0$, the variance still increases by $\frac{\sigma^4}{NI_w^2}$.

\section{Analysis of \AddRWatermark for \pkscheme}
\label{appendix:analysis_of_pk_watermark}
\subsection{Fidelity}
For ciphertext watermark embedding in \pkscheme, we state the following theorem:
\begin{theorem}
\label{th:pk_noise_bound}
When the watermark carrier is a ciphertext of \pkscheme, and the embedding intensity satisfies $I_w < \frac{q}{2pB_e} - 2N - 2$, the homomorphic ciphertext watermarking scheme \AddRWatermark described in \cref{sec:add_watermark} achieves fidelity.
\end{theorem}

\begin{proof}

Let a ciphertext of \pkscheme be
$c_t = (\mathbf{k}_1 \cdot \mathbf{u} + p\mathbf{e}_1 + \mathbf{m},\ \mathbf{k}_2 \cdot \mathbf{u} + p\mathbf{e}_2)$,
where $\mathbf{u} \sim \chi_u$, $\mathbf{e}_1 \sim \chi_e$, and $\mathbf{e}_2 \sim \chi_e$.
After embedding the watermark $w_e$ using \AddRWatermark, the resulting ciphertext is
$c_t^\prime = (\mathbf{k}_1 \cdot \mathbf{u} + p\mathbf{e}_1 + \mathbf{m} + p w_e I_w \mathbf{k}_w,\ \mathbf{k}_2 \cdot \mathbf{u} + p\mathbf{e}_2)$.
Then, the inner product between $c_t^\prime$ and  $\sk$ is given as follows:
\begin{equation}
\label{eq:noise_in_pk}
\begin{aligned}
    \langle\sk,c_t^\prime\rangle&=\mathbf{k}_1\cdot \mathbf{u}+p\mathbf{e}_1+\mathbf{m}+pw_eI_w\mathbf{k}_w +\mathbf{k}_2\cdot \mathbf{u} \cdot \mathbf{s} \\ &+p\mathbf{e}_2\cdot \mathbf{s}\\
    &=(\mathbf{a}_k\cdot \mathbf{s}+p\mathbf{e}_k)\cdot \mathbf{u}+p\mathbf{e}_1+\mathbf{m}+pw_eI_w\mathbf{k}_w\\
    &-\mathbf{a}_k\cdot \mathbf{u} \cdot \mathbf{s}+p\mathbf{e}_2\cdot \mathbf{s} \\
    &=p\mathbf{e}_k\cdot \mathbf{u}+p\mathbf{e}_1+pw_eI_w\mathbf{k}_w+p\mathbf{e}_2\cdot \mathbf{s}+\mathbf{m}.
\end{aligned}  
\end{equation}

By \cref{lm:norm_mul}, we have
$\Vert \mathbf{e}_k \cdot \mathbf{u} \Vert \leq N \Vert \mathbf{e}_k \Vert \cdot \Vert \mathbf{u} \Vert$.
Since $\mathbf{u} \sim \chi_s$ and $\mathbf{s} \sim \chi_s$, it follows that
$\Vert \mathbf{u} \Vert \leq 1$ and $\Vert \mathbf{s} \Vert \leq 1$.
Therefore,
$\Vert \mathbf{e}_k \cdot \mathbf{u} \Vert \leq N \Vert \mathbf{e}_k \Vert$,
and
$\Vert \mathbf{e}_2 \cdot \mathbf{s} \Vert \leq N \Vert \mathbf{e}_k \Vert$.

Hence, we obtain the following inequality:
\begin{equation*}
    \begin{aligned}
     \Vert \langle\sk,c_t^\prime\rangle \Vert
    &=\Vert p\mathbf{e}_k\cdot \mathbf{u}+p\mathbf{e}_1+pw_eI_w\mathbf{k}_w+p\mathbf{e}_2\cdot \mathbf{s} +\mathbf{m} \Vert \\ 
    & \leq pN\Vert \mathbf{e}_k \Vert  + p\Vert \mathbf{e}_1 \Vert+pI_w\Vert \mathbf{k}_w \Vert+pN\Vert \mathbf{e}_2 \Vert +p \\ 
    &\leq 2pNB_e+pI_wB_e+2pB_e.     
    \end{aligned}
\end{equation*}
When $I_w < \frac{q}{2pB_e} - 2N - 2$, we have $\Vert \langle \sk, c_t' \rangle \Vert < \frac{q}{2}$, which implies correct decryption.
\end{proof}

The noise introduced by watermark embedding is mainly affected by $p$, $q$, $I_w$, and $B_e$.
Typically, $p \ll q$, which provides a certain amount of redundancy for watermark embedding and allows $I_w$ to take relatively large values. Furthermore, based on the above analysis, we observe that ciphertexts of \skscheme provide a larger embedding space for watermarking than ciphertexts of \pkscheme.

\subsection{Correctness}

\begin{theorem}
\label{th:pk_e_v}
Let a ciphertext of \pkscheme be
$c_t=(\mathbf{k}_1\cdot \mathbf{u}+p\mathbf{e}_1+\mathbf{m},\mathbf{k}_2\cdot \mathbf{u}+p\mathbf{e}_2)$,
where $\mathbf{u} \sim \chi_u$ and $\mathbf{e}_1,\mathbf{e}_2 \sim \chi_e$.
After embedding the watermark into $c_t$, the expectation of the correlation value $\rho_{pk}$ is $w_e\sigma^2$, and the variance is $\frac{4\sigma^4}{3I_w^2}+\frac{\sigma^4}{NI_w^2}+\frac{2w_e^2\sigma^4}{N}$.
\end{theorem}



\begin{proof}
According to \cref{eq:noise_in_pk} and \cref{eq:extract_noise}, the correlation value extracted from the ciphertext $c_t$ is given by
\begin{equation*}
    \rho_{pk}=\frac{1}{NI_w}\cdot \langle \mathbf{e}_k\cdot \mathbf{u}+\mathbf{e}_1+w_eI_w\mathbf{k}_w+\mathbf{e}_2\cdot \mathbf{s} , \mathbf{k}_w\rangle.
\end{equation*}


By \cref{lm:r_xy}, the expectations of $\mathbf{e}_k\cdot \mathbf{u}$ and $\mathbf{e}_2\cdot \mathbf{s}$ are both $0$, and their variances are both $\frac{2N\sigma^2}{3}$.
Therefore, we have
\begin{equation*}
    \begin{aligned}
        \E(\rho_{pk})&=\E(\frac{1}{NI_w}\cdot \langle \mathbf{e}_k\cdot \mathbf{u}+\mathbf{e}_1+w_eI_w\mathbf{k}_w+\mathbf{e}_2\cdot \mathbf{s} , \mathbf{k}_w \rangle)\\
        &=\frac{1}{NI_w}(\E(\langle\mathbf{e}_k\cdot \mathbf{u},\mathbf{k}_w\rangle)+\E(\langle\mathbf{e}_2\cdot \mathbf{s},\mathbf{k}_w\rangle))\\
        &+\E(\rho_{sk})\\
        &=w_e\sigma^2,    
    \end{aligned}
\end{equation*}
and
\begin{equation*}
\begin{aligned}
    \var(\rho_{pk})&=\var(\frac{1}{NI_w}\cdot \langle \mathbf{e}_k\cdot \mathbf{u}+\mathbf{e}_1+w_eI_w\mathbf{k}_w+\mathbf{e}_2\cdot \mathbf{s} , \mathbf{k}_w \rangle)\\
    &=\frac{1}{N^2I_w^2}(\var(\langle\mathbf{e}_k\cdot \mathbf{u},\mathbf{k}_w\rangle)+\var(\langle\mathbf{e}_2\cdot \mathbf{s},\mathbf{k}_w\rangle))\\
    &+\var(\rho_{sk})\\
    &=\frac{1}{N^2I_w^2}(\frac{2N^2\sigma^4}{3}+\frac{2N^2\sigma^4}{3})+\frac{\sigma^4}{NI_w^2}+\frac{2w_e^2\sigma^4}{N}\\
    &=\frac{4\sigma^4}{3I_w^2}+\frac{\sigma^4}{NI_w^2}+\frac{2w_e^2\sigma^4}{N}.
\end{aligned}
\end{equation*}


\end{proof}

Based on the above analysis, we observe that the variance of the correlation value in \pkscheme is significantly larger than that in \skscheme.
This implies that, to make the extracted correlation values more concentrated, \pkscheme requires a larger embedding intensity $I_w$.

\subsection{Robustness}
\begin{theorem}
\label{th:pk_r}
Let the \pkscheme ciphertext 
$c_{t,0}=(\mathbf{k}_1\cdot \mathbf{u}_0+p\mathbf{e}_1+\mathbf{m}_0+pw_eI_w\mathbf{k}_w,\mathbf{k}_2 \cdot \mathbf{u}_0+p\mathbf{e}_2)$ 
be the ciphertext embedded with watermark $w_e$ with embedding key $\mathbf{k}_w$, and let 
$c_{t,1}=(\mathbf{k}_1\cdot \mathbf{u}_1+p\mathbf{e}_3+\mathbf{m}_1+pw_e^\prime I_w^\prime \mathbf{k}_w^\prime,\mathbf{k}_2 \cdot \mathbf{u}_1+p\mathbf{e}_4)$ be the ciphertext embedded with watermark $w_e^\prime$ with embedding key $\mathbf{k}_w^\prime$. 
The embedding key $\mathbf{k}_w$ is independent of $\mathbf{k}_w^\prime$.  
Let $\rho_{pk}$ denote the correlation value of $c_{t,0}$ and $\rho^\prime_{pk}$ denote the correlation value computed from the ciphertext $c_{t,0}+c_{t,1}$ when taking $\mathbf{k}_w$ as extracting key. Then,
$\E(\rho_{pk})=\E(\rho^\prime_{pk}),$
and
$\var(\rho^\prime_{pk})=\frac{\sigma^4}{NI_w^2}+\frac{w_e^{\prime2} I_w^{\prime2}\sigma^4}{NI_w^2}+\frac{4\sigma^4}{3I_w^2}+\var(\rho_{pk}).$
\end{theorem}

\begin{proof}
After extracting the noise from $c_{t,0}+c_{t,1}$ using \cref{eq:extract_noise}, we obtain
\begin{equation*}
    \begin{aligned}
    \rho&=\frac{1}{NI_w}\cdot \langle \mathbf{e}_k \cdot \mathbf{u}_0+\mathbf{e}_1+w_eI_w\mathbf{k}_w+\mathbf{e}_2\cdot \mathbf{s}\\
    &+\mathbf{e}_k \cdot \mathbf{u}_1+\mathbf{e}_3+w_e^\prime I_w^\prime \mathbf{k}_w^\prime+\mathbf{e}_4\cdot \mathbf{s},\mathbf{k}_w\rangle.
    \end{aligned}
\end{equation*}

Because $\mathbf{k}_w^\prime$ and $\mathbf{k}_w$ are independent and the have zero expectation, we have
\begin{equation*}
    \begin{aligned}
    \E(\rho^\prime_{pk})&=\frac{1}{NI_w}\E(\langle \mathbf{e}_k \cdot \mathbf{u}_0+\mathbf{e}_1+w_eI_w\mathbf{k}_w+\mathbf{e}_2\cdot \mathbf{s}\\
    &+\mathbf{e}_k \cdot \mathbf{u}_1+\mathbf{e}_3+w_e^\prime I_w^\prime \mathbf{k}_w^\prime+\mathbf{e}_4\cdot \mathbf{s},\mathbf{k}_w\rangle)\\
    &=\frac{1}{NI_w}\E(\langle\mathbf{e}_k \cdot \mathbf{u}_1+\mathbf{e}_3+w_e^\prime I_w^\prime \mathbf{k}_w^\prime+\mathbf{e}_4\cdot \mathbf{s},\mathbf{k}_w\rangle)\\
    &+\E(\rho_{pk})\\
    &=\E(\rho_{pk}),    
    \end{aligned}
\end{equation*}
and
\begin{equation*}
    \begin{aligned}
    \var(\rho^\prime_{pk})&=\frac{1}{N^2I_w^2}\var(\langle \mathbf{e}_k \cdot \mathbf{u}_0+\mathbf{e}_1+w_eI_w\mathbf{k}_w+\mathbf{e}_2\cdot \mathbf{s}\\
    &+\mathbf{e}_k \cdot \mathbf{u}_1+\mathbf{e}_3+w_e^\prime I_w^\prime \mathbf{k}_w^\prime+\mathbf{e}_4\cdot \mathbf{s},\mathbf{k}_w\rangle)\\
    &=\frac{1}{N^2I_w^2}\var(\langle \mathbf{e}_k \cdot \mathbf{u}_1+\mathbf{e}_3+w_e^\prime I_w^\prime \mathbf{k}_w^\prime+\mathbf{e}_4\cdot \mathbf{s},\mathbf{k}_w\rangle)\\
    &+\var(\rho_{pk})\\
    &=\frac{\sigma^4}{NI_w^2}+\frac{w_e^{\prime2} I_w^{\prime2}\sigma^4}{NI_w^2}+\frac{4\sigma^4}{3I_w^2}+\var(\rho_{pk}).    
    \end{aligned}
\end{equation*}

\end{proof}

Therefore, for \AddRWatermark watermarked \pkscheme ciphertexts, adding another ciphertext embedded with watermark $w_e^\prime$ increases the variance of $\rho$ by
$
\frac{\sigma^4}{NI_w^2}+\frac{w_e^{\prime2} I_w^{\prime2}\sigma^4}{NI_w^2}+\frac{4\sigma^4}{3I_w^2}.
$
As the variance increases, the distributions of $\rho$ corresponding to different values of $w_e$ are more likely to overlap, which leads to a lower TPR. However, the robustness of the watermark can be improved by increasing the dimension $N$ and the embedding intensity $I_w$.

\end{document}